\documentclass[aps,prx,floatfix,nofootinbib,superscriptaddress,twocolumn,longbibliography,10pt,raggedbottom]{revtex4-2}
\pdfoutput=1

\usepackage{xcolor}
\usepackage{times}
\usepackage{graphicx}
\usepackage{xcolor}
\usepackage{xr-hyper}

\let\coloneqq\relax

\usepackage[T1]{fontenc}

\usepackage{amsthm}
\usepackage{amssymb}
\usepackage{amsmath}
\usepackage{bbold}
\usepackage{bbm}
\usepackage{nonfloat}
\usepackage{xcolor}
\definecolor{myrefcolor}{rgb}{0.067,0.5,0.5}
\usepackage[
    breaklinks,
    pdftex,
    colorlinks=true,
    linkcolor=myrefcolor,
    citecolor=myrefcolor,
    urlcolor=myrefcolor
]{hyperref}
\usepackage{braket}
\usepackage{dsfont}
\usepackage{mathdots}
\usepackage{mathtools}
\usepackage{enumerate}
\usepackage[shortlabels]{enumitem}
\usepackage{csquotes}
\usepackage{stmaryrd}
\usepackage[cal=boondox]{mathalfa}
\usepackage{graphicx}
\usepackage{stackengine}
\usepackage{scalerel}
\usepackage{microtype}
\usepackage{array}
\usepackage{makecell}
\newcolumntype{x}[1]{>{\centering\arraybackslash}p{#1}}
\usepackage{tikz}
\usepackage{pgfplots}
\usetikzlibrary{shapes.geometric, shapes.misc, positioning, arrows, arrows.meta, decorations.pathreplacing, decorations.pathmorphing, patterns, angles, quotes, calc}
\usepackage{booktabs}
\usepackage{xfrac}
\usepackage{siunitx}
\usepackage{centernot}
\usepackage{comment}
\usepackage{chngcntr}
\usepackage{breakurl}
\def\clearthms#1{ \@for\tname:=#1\do{\cleartheorem\tname} }

\newtheorem*{thm*}{Theorem}
\newtheorem{thm}{Theorem}
\newtheorem{prop}[thm]{Proposition}
\newtheorem*{prop*}{Proposition}
\newtheorem{lemma}[thm]{Lemma}
\newtheorem*{lemma*}{Lemma}

\newtheorem*{cor*}{Corollary}

\newtheorem*{cj*}{Conjecture}

\newtheorem*{Def*}{Definition}

\newtheorem{remark}[thm]{Remark}

\makeatletter
\def\thmhead@plain#1#2#3{%
  \thmname{#1}\thmnumber{\@ifnotempty{#1}{ }\@upn{#2}}%
  \thmnote{ {\the\thm@notefont#3}}}
\let\thmhead\thmhead@plain
\makeatother

\theoremstyle{definition}

\newcommand{\bb}{\begin{equation}\begin{aligned}\hspace{0pt}}
\newcommand{\be}{\begin{equation}\begin{aligned}\hspace{0pt}}
\newcommand{\bbb}{\begin{equation*}\begin{aligned}}
\newcommand{\ee}{\end{aligned}\end{equation}}
\newcommand{\eee}{\end{aligned}\end{equation*}}
\newcommand*{\coloneqq}{\mathrel{\vcenter{\baselineskip0.5ex \lineskiplimit0pt \hbox{\scriptsize.}\hbox{\scriptsize.}}} =}

\newcommand{\eqt}[1]{\stackrel{\mathclap{\scriptsize \mbox{#1}}}{=}}
\newcommand{\leqt}[1]{\stackrel{\mathclap{\scriptsize \mbox{#1}}}{\leq}}

\newcommand{\ketbra}[1]{\ket{#1}\!\!\bra{#1}}

\newcommand{\id}{\mathds{1}}

\newcommand{\C}{\mathds{C}}

\DeclareMathOperator{\Tr}{Tr}

\DeclareMathAlphabet{\pazocal}{OMS}{zplm}{m}{n}

\newcommand{\HH}{\pazocal{H}}

\newcommand{\lsmatrix}{\left(\begin{smallmatrix}}
\newcommand{\rsmatrix}{\end{smallmatrix}\right)}

\stackMath

\stackMath

\makeatletter
\newcommand*\rel@kern[1]{\kern#1\dimexpr\macc@kerna}
\newcommand*\widebar[1]{%
  \begingroup
  \def\mathaccent##1##2{%
    \rel@kern{0.8}%
    \overline{\rel@kern{-0.8}\macc@nucleus\rel@kern{0.2}}%
    \rel@kern{-0.2}%
  }%
  \macc@depth\@ne
  \let\math@bgroup\@empty \let\math@egroup\macc@set@skewchar
  \mathsurround\z@ \frozen@everymath{\mathgroup\macc@group\relax}%
  \macc@set@skewchar\relax
  \let\mathaccentV\macc@nested@a
  \macc@nested@a\relax111{#1}%
  \endgroup
}

\counterwithin*{equation}{part}
\counterwithin*{thm}{part}
\counterwithin*{figure}{part}

\tikzset{meter/.append style={draw, inner sep=10, rectangle, font=\vphantom{A}, minimum width=30, line width=.8, path picture={\draw[black] ([shift={(.1,.3)}]path picture bounding box.south west) to[bend left=50] ([shift={(-.1,.3)}]path picture bounding box.south east);\draw[black,-latex] ([shift={(0,.1)}]path picture bounding box.south) -- ([shift={(.3,-.1)}]path picture bounding box.north);}}}
\tikzset{roundnode/.append style={circle, draw=black, fill=gray!20, thick, minimum size=10mm}}
\tikzset{squarenode/.style={rectangle, draw=black, fill=none, thick, minimum size=10mm}}

\definecolor{Blues5seq1}{RGB}{239,243,255}
\definecolor{Blues5seq2}{RGB}{189,215,231}
\definecolor{Blues5seq3}{RGB}{107,174,214}
\definecolor{Blues5seq4}{RGB}{49,130,189}
\definecolor{Blues5seq5}{RGB}{8,81,156}

\definecolor{Greens5seq1}{RGB}{237,248,233}
\definecolor{Greens5seq2}{RGB}{186,228,179}
\definecolor{Greens5seq3}{RGB}{116,196,118}
\definecolor{Greens5seq4}{RGB}{49,163,84}
\definecolor{Greens5seq5}{RGB}{0,109,44}

\definecolor{Reds5seq1}{RGB}{254,229,217}
\definecolor{Reds5seq2}{RGB}{252,174,145}
\definecolor{Reds5seq3}{RGB}{251,106,74}
\definecolor{Reds5seq4}{RGB}{222,45,38}
\definecolor{Reds5seq5}{RGB}{165,15,21}

\usepackage{pgfplots}

\newtheorem{example}{Example}

\newcommand*{\addFileDependency}[1]{
  \typeout{(#1)}
  \@addtofilelist{#1}
  \IfFileExists{#1}{}{\typeout{No file #1.}}
}

\makeatother

\definecolor{tealblue}{HTML}{00AEB3}

\pgfplotsset{width=10cm,compat=1.9}
\usetikzlibrary{decorations.markings}

\newcommand{\nocontentsline}[3]{}
\let\origcontentsline\addcontentsline
\newcommand\stoptoc{\let\addcontentsline\nocontentsline}
\newcommand\resumetoc{\let\addcontentsline\origcontentsline}

\PassOptionsToPackage{pdftex,hyperfootnotes=false,pdfpagelabels}{hyperref}
\hypersetup{
	colorlinks = true,
	linkcolor = [rgb]{0.70,0.13,0.13},
	citecolor = [rgb]{0.13,0.55,0.13},
	urlcolor = [rgb]{0.25, 0.41, 0.88}}

\usepackage{amsmath,amssymb,amsthm}
\usepackage{physics}
\usepackage{bm}
\usepackage{graphicx}
\usepackage{xcolor}

\usepackage{float}              
\usepackage{algorithm}
\usepackage[noend]{algpseudocode}

\newcommand{\fu}{\small Dahlem Center for Complex Quantum Systems, Freie Universit\"{a}t Berlin, 14195 Berlin, Germany}

\newcommand{\pis}{NEST, Scuola Normale Superiore and Istituto Nanoscienze, Piazza dei Cavalieri 7, IT-56126 Pisa, Italy}

\definecolor{antonio}{rgb}{.2,.5,.1}

\usepackage{amsmath}

\renewcommand{\ketbra}[2]{\ket{#1}\!\bra{#2}}
\usepackage{times}

\allowdisplaybreaks



\begin{document}
\title{Energy-independent tomography of Gaussian states}
\author{Lennart Bittel}
\thanks{\href{mailto:l.bittel@fu-berlin.de}{l.bittel@fu-berlin.de}}
\affiliation{\fu}

\author{Francesco A. Mele}
\thanks{\href{mailto:francesco.mele@sns.it}{francesco.mele@sns.it}}
\affiliation{\pis}

\author{Jens Eisert}
\thanks{\href{mailto:jense@zedat.fu-berlin.de}{jense@zedat.fu-berlin.de}}
\affiliation{\fu}

\author{Antonio A. Mele}
\thanks{\href{mailto:a.mele@fu-berlin.de}{a.mele@fu-berlin.de}}
\affiliation{\fu}

\begin{abstract}
{Abstract:} The exploration of tomography of bosonic Gaussian states is presumably as old as quantum optics, but only recently, their precise and rigorous study have been moving into the focus of attention, motivated by technological developments. In this work, we present an efficient and experimentally feasible Gaussian state tomography algorithm with provable recovery trace-distance guarantees, whose sample complexity depends only on the number of modes, and---remarkably---is independent of the state's photon number or energy, up to doubly logarithmic factors. Our algorithm yields a doubly-exponential improvement over existing methods, and it employs operations that are readily accessible in experimental settings: the preparation of an auxiliary squeezed vacuum, passive Gaussian unitaries, and homodyne detection. At its core lies an adaptive strategy that systematically reduces the total squeezing of the system, enabling efficient tomography. Quite surprisingly, this proves that estimating a Gaussian state in trace distance is generally more efficient than directly estimating its covariance matrix. 
Our algorithm is particularly well-suited for applications in quantum metrology and sensing, where highly squeezed---and hence high-energy---states are commonly employed. As a further contribution, we establish improved sample complexity bounds for standard heterodyne tomography, equipping this widely used protocol with rigorous trace-norm guarantees.
\end{abstract}

\hypersetup{%
       pdftitle = {Energy-independent tomography of Gaussian states},
       pdfauthor = {Lennart Bittel, Francesco A. Mele, Jens Eisert and Antonio A. Mele},
       pdfsubject = {Gaussian tomography},
       pdfkeywords = {Tomography, Gaussian states, energy independent, Trace distance, one norm, sampling efficient, squeezing, coherent, thermal, heterodyne, homodyne, generaldyne, quantum optics, continuous-variable quantum information}, 
      }

\maketitle
\stoptoc

\section*{Introduction}
Quantum state tomography~\cite{cramer_efficient_2010, haah_optimal_2021,odonnell_efficient_2016}---the task of reconstructing an unknown quantum state from measurement data---is a fundamental primitive in quantum physics. It provides the most detailed form of system identification and underpins a wide range of applications, including benchmarking and device calibration~\cite{eisert_quantum_2020,kliesch_theory_2021}. In quantum optics specifically, the need to characterize states of light prepared in the laboratory has been central since the field’s inception~\cite{smithey_measurement_1993,leonhardt_measuring_1995,lvovsky_continuousvariable_2009}.

With the rapid progress of quantum technologies---from qubit-based to continuous variable photonic platforms \cite{Fusion,Xanadu,alexander2024manufacturable,Borealis,Zhong_2020,QuantumPhotoThermodynamics}---there has recently been a growing demand for tomographic protocols that are not only experimentally feasible but also come with rigorous, quantitative recovery guarantees~\cite{anshu2023survey}. For many applications of quantum state tomography~\cite{anshu2023survey}, the relevant metric for quantifying reconstruction accuracy is often the \emph{trace distance}~\cite{NC,MARK,KHATRI} between the reconstructed and true states, widely regarded as the most operationally meaningful way to distinguish quantum states~\cite{HELSTROM, Holevo1976}. While substantial progress in this context has been made in finite-dimensional settings---including efficient protocols for structured families of quantum states~\cite{Cramer_2010,Lanyon_2017,grewal2023efficient,leone2023learning,hangleiter2024bell,mele2024learningquantumstatescontinuous, bittel2024optimalestimatestracedistance,fanizza2023learning,mele2024efficient,arunachalam2023optimal,iyer2025mildlyinteractingfermionicunitariesefficiently,austin2025efficientlylearningfermionicunitaries,montanaro2017learning,rouzé2023learning,landau2024learningquantumstatesprepared,kim2024learningstatepreparationcircuits,bittel2024optimaltracedistanceboundsfreefermionic}---only recently has attention shifted toward \emph{continuous-variable} (CV) systems, where the infinite-dimensional Hilbert space introduces qualitatively new conceptual and technical challenges~\cite{mele2024learningquantumstatescontinuous, bittel2024optimalestimatestracedistance, fanizza2024efficienthamiltonianstructuretrace, holevo2024estimatestracenormdistancequantum, gandhari_precision_2023, becker_classical_2023, oh2024entanglementenabled, liu2025quantumlearningadvantagescalable,coroi2025exponentialadvantagecontinuousvariablequantum,fawzi2024optimalfidelityestimationbinary, wu2024efficient,möbus2023dissipationenabledbosonichamiltonianlearning,upreti2024efficientquantumstateverification,zhao2025complexityquantumtomographygenuine}.
Importantly, physical CV states are never truly arbitrary infinite-dimensional quantum states---they are constrained by physical resources, implying finite mean energy or total photon number~\cite{BUCCO}. 

A particularly important subclass of CV states is constituted
by \emph{Gaussian states}, which are Gibbs states of quadratic Hamiltonians in the bosonic quadrature operators and completely characterized by their first and second moments (i.e., their displacement vector and covariance matrix)~\cite{BUCCO,weedbrook12}. These states are ubiquitous in quantum optics, as they naturally emerge in a wide range of experimental settings, including squeezed light generation, thermal noise modeling, and interferometric setups, such as those used in gravitational wave detectors~\cite{Schnabel_2017,Schnabel2010,SGravi2,PhysRevLett.121.160502, PhysRevLett.86.5870,q_sensing_cv,nguyen2024digitalreconstructionsqueezedlight,BUCCO}. They also play a central role in quantum communication, sensing, and simulation~\cite{RennerPhD,weedbrook12,Lloyd1999,Gottesman2001,Menicucci2006,Mirrahimi_2014,Ofek_nature2016,error_corr_boson,Guillaud_2019,alexander2024manufacturable,holwer,Wolf2007,TGW,PLOB,Die-Hard-2-PRL,mele2023maximum, mele2024quantum,Die-Hard-2-PRA,mele2023optical,Borealis,Zhong_2020,SupremacyReview}. Tomography of Gaussian states---typically performed via homodyne or heterodyne detection to estimate the state’s first and second moments---is standard practice in CV experiments~\cite{BUCCO,teo2020highlyaccurategaussianprocess,Esposito_2014,Kumar_2022,kawasaki2024highrategenerationstatetomography}. However, existing experimental procedures often lack rigorous error bounds in trace distance and, crucially, their sample complexity tends to scale poorly with the state’s energy~\cite{mele2024learningquantumstatescontinuous}.

In this work, we first equip the standard heterodyne tomography protocol for Gaussian states---routinely implemented in optical laboratories around the world~\cite{Schnabel_2017,nguyen2024digitalreconstructionsqueezedlight, _eh_ek_2015,teo2020highlyaccurategaussianprocess,Esposito_2014,Kumar_2022,kawasaki2024highrategenerationstatetomography}---with rigorous recovery guarantees in trace distance. Specifically, we show that heterodyne measurements performed on \( \mathcal{O}(nE^2/\varepsilon^2) \) copies of an unknown \( n \)-mode Gaussian state with mean total energy \( E \) suffice to reconstruct the state to accuracy \( \varepsilon \) in trace distance. Importantly, this result establishes a simple and experimentally accessible criterion for successful tomography using the most widely adopted measurement scheme in quantum optics, thereby offering practical guidance to experimentalists. Furthermore, it improves upon the previous state-of-the-art sample complexity of \( \mathcal{O}(nE^4/\varepsilon^2) \)~\cite{bittel2024optimalestimatestracedistance,fanizza2024efficienthamiltonianstructuretrace}. Notably, this shows that estimating Gaussian states in trace distance scales similar to estimating its covariance matrix in trace distance.

Strikingly, we go beyond this standard paradigm. To the best of our knowledge, all existing tomography protocols exhibit sample complexity that scales with the state’s energy~\cite{mele2024learningquantumstatescontinuous,bittel2024optimalestimatestracedistance,fanizza2024efficienthamiltonianstructuretrace}, rendering them increasingly impractical in high-energy regimes---for instance, when dealing with highly squeezed states, which play a central role in cutting-edge experiments in fundamental physics~\cite{Schnabel_2017,Schnabel2010,SGravi2,PhysRevLett.121.160502, PhysRevLett.86.5870,q_sensing_cv,nguyen2024digitalreconstructionsqueezedlight,BUCCO}. Yet, no fundamental lower bound precludes the possibility of circumventing this energy dependence for structured families of states such as Gaussian states~\cite{mele2024learningquantumstatescontinuous}.

This raises a natural question---one whose positive resolution would lead to immediate and practical resource savings in routine Gaussian tomography experiments~\cite{Schnabel_2017,_eh_ek_2015,teo2020highlyaccurategaussianprocess,Esposito_2014,Kumar_2022,kawasaki2024highrategenerationstatetomography,nguyen2024digitalreconstructionsqueezedlight}:
\begin{quote}
\centering
\emph{Can we design a Gaussian state tomography algorithm whose sample complexity is independent of the state's energy?}
\end{quote}
In this work, we answer this question effectively in the affirmative. We present an adaptive and experimentally feasible Gaussian state tomography algorithm that relies exclusively on standard tools available in optical laboratories: auxiliary squeezed state preparation, passive Gaussian operations, and homodyne detection. If the mean total energy of the unknown state is denoted by \( E \), we show that the sample complexity of our protocol scales as \( \log\log(E) \)---a quantity that remains effectively constant (on the order of 10), even when \( E \) is as large as the energy of the observable universe. In contrast, previous state-of-the-art protocols exhibit sample complexities that scale polynomially with \( E \); for instance, standard heterodyne tomography based solely on covariance matrix estimation scales as \( \mathcal{O}(E^2) \), as discussed earlier. Our method thus achieves a double-exponential improvement in energy dependence, while retaining the scaling in the number of modes \( n \). This results in a substantial reduction in experimental cost for high-energy states.

The core idea behind our protocol is to first determine---using only a small number of measurements---whether significant squeezing is initially present in the unknown state. If so, we adaptively suppress this squeezing by applying passive Gaussian unitaries in combination with suitably prepared auxiliary squeezed inputs, effectively implementing generalized heterodyne (or generaldyne) measurements~\cite{BUCCO}. Once the effective squeezing is sufficiently reduced, we perform standard covariance estimation in the rotated basis, and then invert the transformation to reconstruct the original state.
Importantly, this protocol relies exclusively on \emph{offline} squeezing---that is, squeezing applied to auxiliary vacuum states prior to the measurement process---as is routinely available in quantum optics laboratories. At no point is \emph{online} or active squeezing required, making the scheme both experimentally accessible and practically robust.

Moreover, we show that even the mild \( \log\log(E) \) dependence can be removed entirely---yielding a fully energy-independent tomography scheme---if access to the transposed state is available. In this setting, a simple and experimentally feasible procedure based on double homodyne measurements on the state and its transpose suffices. Recently, considerable attention has been devoted to the fact that access to the transposed state can yield substantial advantages in a variety of quantum learning tasks~\cite{King_2024,montanaro2017learningstabilizerstatesbell,Gross_2021,Haug_2024,Haug_2025,Miyazaki_2022,Miyazaki_2024,grewal2024improvedstabilizerestimationbell,wu2024efficient}. Our work contributes to this growing line of research by demonstrating that a similar \emph{transpose-induced} quantum advantage also arises in the tomography of Gaussian states.

These results are grounded in several technical contributions of independent interest. First, we derive a new perturbation bound for the trace distance between Gaussian states in terms of their first two moments, featuring a favourable functional dependence on the state's energy. This advances the rapidly developing toolkit of trace distance bounds for bosonic systems~\cite{bittel2024optimalestimatestracedistance,mele2024learningquantumstatescontinuous,holevo2024estimatestracenormdistancequantum,fanizza2024efficienthamiltonianstructuretrace,mele2025achievableratesnonasymptoticbosonic,mele2025symplecticranknongaussianquantum,holevo2024estimatesburesdistancebosonic}. Second, we conduct a careful sample complexity analysis, leveraging tools from statistical learning theory~\cite{vershynin2020high} on the estimation of covariance matrices from samples of classical Gaussian distributions~\cite{BUCCO}.


\begin{figure*}[t!]
    \centering
    \includegraphics[width=0.99\linewidth]{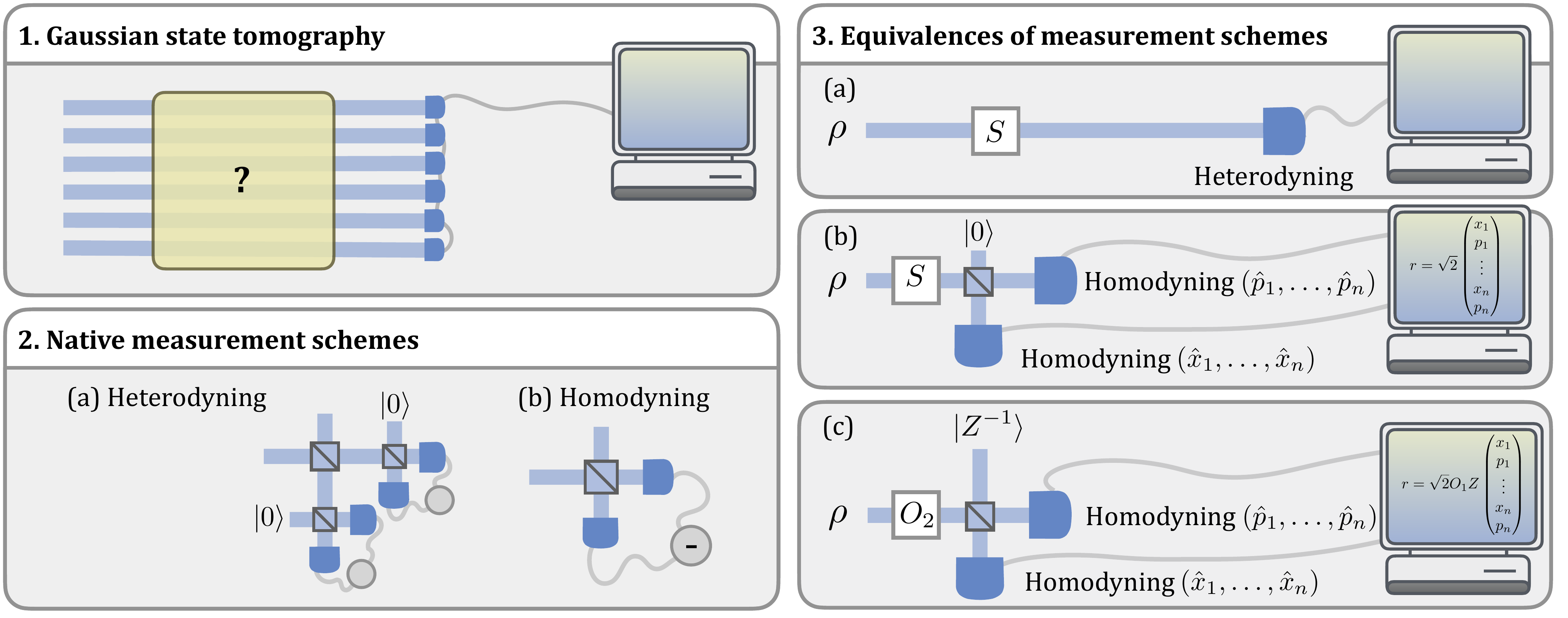}
    \caption{
    This work focuses on Gaussian quantum state tomography: 
    {1.} 
    Reconstructing an unknown bosonic Gaussian state from measurement data, with rigorous trace-distance guarantees.
    2. Common measurements in this setting include experimentally accessible heterodyne (double homodyne) and homodyne detection.
    3. Our adaptive protocol avoids online squeezing by relying solely on passive linear optics and offline squeezing. It simulates the active Gaussian unitary associated with the symplectic matrix \( S = O_1 Z O_2 \) (Euler decomposition) by applying passive transformations to the input state, injecting an auxiliary squeezed vacuum \( \ket{Z^{-1}}\coloneqq U_{Z^{-1}} \ket{0} \) (with \( U_{Z^{-1}} \) the Gaussian unitary implementing the inverse squeezing)
    featuring offline squeezing, 
    followed by homodyne detection and classical post-processing.}
    \label{fig:BigPlot}
\end{figure*}




{We continue by stating some preliminaries.}
The trace distance 
between quantum states 
\( \rho \) and 
\( \sigma \) 
is defined 
in terms of the trace
norm \( 
\|\cdot\|_1 
\)
as 
\cite{wilde_quantum_2013,WATROUS}
\begin{align}
    D_{\mathrm{tr}}(\rho, \sigma) \coloneqq \frac{1}{2} \|\rho - \sigma\|_1.
\end{align} 
This metric is operationally meaningful, as it determines the optimal success probability in distinguishing \( \rho \) from \( \sigma \) in a single-shot measurement, according to the Holevo–Helstrom theorem~\cite{BUCCO}.

We now introduce some relevant notation for CV systems; see the {supplemental material} or Refs.~\cite{BUCCO,weedbrook12} for further details. A CV system with \( n \) bosonic modes is described on \( L^2(\mathbb{R}^n) \), with canonical quadratures 
\( \hat{\mathbf{R}} \coloneqq (\hat{x}_1, \hat{p}_1, \dots, \hat{x}_n, \hat{p}_n)^\intercal \) satisfying \( [\hat{R}_j, \hat{R}_k] = i\,\Omega_{j,k} \), where 
\(\Omega \coloneqq \bigoplus_{i=1}^n \begin{psmallmatrix} 0 & 1 \\ -1 & 0 \end{psmallmatrix} \) is the symplectic form. The total energy operator is \( \hat{E}_n \coloneqq \sum_{j=1}^n \left( {\hat{x}_j^2} + {\hat{p}_j^2} \right) /2\), and a state has bounded mean energy \( {\rm Tr}(\rho \hat{E}_n) \le E \).
Gaussian states are Gibbs states of quadratic Hamiltonians in the quadratures. They are fully specified by the first moment \( m(\rho) \coloneqq 
{\rm Tr}[\hat{\mathbf{R}} \rho] \) and the covariance matrix \( V(\rho) \), which encodes second moments and satisfies the uncertainty relation \( V + i \Omega \ge 0 \)~\cite{BUCCO}. Any such \( V \) admits a \emph{Williamson decomposition} \( V = S D S^\intercal \), with \( S \in \mathrm{Sp}(2n) \) symplectic and \( D = \mathrm{diag}(\nu_1, \nu_1, \dots, \nu_n, \nu_n) \) containing the symplectic eigenvalues.
Gaussian unitaries are unitary operations generated by quadratic Hamiltonians in the quadratures. \emph{Active} and \emph{passive} Gaussian unitaries are fully specified by their associated symplectic matrices \( S \in \mathrm{Sp}(2n) \). A Gaussian unitary is called \emph{passive} if \( S \in \mathrm{O}(2n) \cap \mathrm{Sp}(2n) \); otherwise, it is called \emph{active}. Passive operations---like beam splitters and phase shifters---can be implemented easily with linear optics, whereas active ones, such as squeezing, are more experimentally demanding. Any \( S \in \mathrm{Sp}(2n) \) admits an Euler decomposition \( S = O_1 Z O_2 \), where \( O_{1},O_{2} \in \mathrm{Sp}(2n) \cap \mathrm{O}(2n) \), and \( Z = \mathrm{diag}(z_1, 1/z_1, \dots, z_n, 1/z_n) \) with \( z_j \geq 1 \).

Standard measurement strategies in CV systems routinely used in the lab include \emph{homodyne}, \emph{heterodyne}, and more generally, \emph{generaldyne} detection. Homodyne 
measurements on Gaussian states correspond to sampling from a classical Gaussian distribution determined by specific quadrature operators and the corresponding marginal of the state covariance matrix. Heterodyne detection, on the other hand, simultaneously probes both quadratures with added noise---effectively sampling from a Gaussian distribution with the state covariance matrix shifted by the identity. In practice, heterodyne detection is implemented by interfering the input state with a vacuum auxiliary state at a balanced beam splitter, followed by homodyne measurements on both output ports. More generally, \emph{generaldyne} measurements form a broad class of Gaussian POVMs---including homodyne and heterodyne as special cases---which can be realized by combining the input state with an auxiliary Gaussian state via a passive Gaussian unitary, followed by homodyne detection.


\section*{Results}

\subsection*{Tomography using only heterodyne measurements}

\begin{algorithm}[H]
\caption[Heterodyne tomography of Gaussian states]{Heterodyne tomography of Gaussian states}
\label{algo_het}
\begin{algorithmic}[1]
\State \textbf{Input}: $N_t$ copies of the Gaussian state $\rho$, failure probability $\delta$
\For{$j \in [N_t]$}
    \State $\hat r_j \gets \mathrm{Heterodyne}(\rho)$
\EndFor
\State $\hat{m} \gets \frac{1}{N_t} \sum_{j=1}^{N_t} \hat r_j$
\State $\hat{V} \gets \frac{{2}}{N_t(1-\zeta)} \sum_{j=1}^{N_t} (\hat r_j - \hat m)(\hat r_j - \hat m)^\top - \mathbb{1}$
\State \textbf{Return}: $\hat{m}$, $\hat{V}$
\end{algorithmic}
where $\zeta\coloneqq\frac{2\chi}{\sqrt{N_t}} + \frac{2\chi^2}{N_t}$ 
and $\chi \coloneqq \sqrt{2n} + \sqrt{2 \log(2/\delta)}$.
\end{algorithm}

We begin by analyzing the performance guarantees---in trace distance---of the heterodyne tomography protocol routinely implemented in optical laboratories~\cite{BUCCO,weedbrook12}.
Our result provides this standard method with rigorous recovery guarantees, linking experimental practice to trace distance bounds often required in applications. See Theorem~10 in the {supplemental material} for explicit constants and the full proof.
\begin{thm}[{\normalfont (Non-adaptive heterodyne tomography, informal; see Thm.~10)}]
Let \( \varepsilon, \delta \in (0,1) \), and let \( \rho \coloneqq \rho(V, m) \) be an \( n \)-mode Gaussian state with covariance matrix \( V \) and first moment \( m \).  
Then, by performing heterodyne measurements and reconstructing the empirical first moment \( \hat{m} \) and covariance matrix \( \hat{V} \) as in Algorithm~\ref{algo_het}, we obtain an estimated state \( \hat{\rho} \coloneqq \rho(\hat{V}, \hat{m}) \) satisfying
\begin{align}
   D_{\mathrm{tr}}(\rho, \hat{\rho}) \leq 4.3 \left( 2n + \Tr(V^{-1}) \right) \frac{\chi}{\sqrt{N}},
\end{align}
with probability at least \( 1 - \delta \), where \( \chi = \sqrt{2n} + \sqrt{2 \log(2/\delta)} \).  
In particular, to ensure \( D_{\mathrm{tr}}(\rho, \hat{\rho}) \le \varepsilon \) with probability at least \( 1 - \delta \), it suffices to use
\begin{align}
   N = O\left( \frac{n E^2}{\varepsilon^2} \right)
\end{align}
measurement shots, where \( E \) is the energy of the state.
\end{thm}

This improves upon the previously best-known bounds of \( O(n E^4 / \varepsilon^2) \)~\cite{bittel2024optimalestimatestracedistance,fanizza2024efficienthamiltonianstructuretrace}. Specifically, our sample complexity scales as
\begin{align}
N = O\left( \frac{n^3+n\, \Tr^2(V^{-1})}{\varepsilon^2} \right),
\end{align}
where the trace term satisfies \( \Tr^2(V^{-1}) \le 4n^2 \lambda_{\min}^{-2}(V) \), with \( \lambda_{\min}(V) \) denoting the smallest eigenvalue of \( V \).

This implies that the energy dependence in the bound above arises solely from squeezing. For non-squeezed states---such as thermal or coherent states, or any state preparable via passive Gaussian unitaries---the scaling improves to \( O(n^3 / \varepsilon^2) \), as \( \lambda_{\min}(V) \) remains lower bounded by a constant.

The key observation underlying this theorem is that heterodyne measurements on a Gaussian state \( \rho(V, m) \) yield samples from a classical Gaussian distribution~\cite{BUCCO} with mean \( m \) and covariance \( (V + \mathbb{1})/2 \). This allows us to apply standard concentration bounds to estimate the first and second moments, which in turn translate into a trace distance bound between \( \rho \) and its estimate \( \hat{\rho} \) (the Gaussian state associated with the estimated moments).

To this end, we derive a scale-invariant perturbation bound---of independent interest---that quantifies how closeness in moments implies closeness in trace distance for Gaussian states.
\begin{thm}[{\normalfont (Perturbation bound for Gaussian states; see Thm.~8)}]
\label{thm_upp_bound_gFIRST}
Let \( \rho(V, m) \) and \( \rho(W, t) \) be \( n \)-mode Gaussian states with first moments \( m, t \in \mathbb{R}^{2n} \) and covariance matrices \( V, W \in \mathbb{R}^{2n \times 2n} \).
Then the trace distance between the two states satisfies
\begin{samepage}
\begin{align}
    D_{\mathrm{tr}}(\rho(V, m), \rho(W, t)) 
    &\le \left\| V^{-1/2} (m - t) \right\|_2  \\
    &+ \frac{1\!+ \!\sqrt{3}}{8} \Tr\!\left[ (V^{-1}\! +\! W^{-1}) |V\! -\! W| \right], \nonumber
\end{align}
where $|A|=\sqrt{\!A^2}$ refers to the matrix absolute value for Hermitian matrices.
\end{samepage}
\end{thm}

The proof of this theorem is provided in the {supplemental material} (Theorem~8).



\subsection*{Adaptive and effectively energy-independent tomography}
The core result of this work is a tomography protocol for Gaussian states whose sample complexity is effectively independent of the state's energy. The scheme is formally described in Algorithm~\ref{algo_ad}. Its performance guarantees are summarized below, and the proof is given in Theorem~14 in the {supplemental material}.
\begin{thm}[{\normalfont (Adaptive Gaussian state tomography; see Thm.~14)}]
\label{th:adapTOMmain}
Let \( \varepsilon, \delta \in (0,1) \), and let \( \rho \) be an unknown \( n \)-mode Gaussian state with covariance matrix \( V \).
Then, by performing the adaptive sampling strategy described in Algorithm~\ref{algo_ad}, one obtains an estimate \( \hat{\rho} \) such that \( D_{\mathrm{tr}}(\rho, \hat{\rho}) \le \varepsilon \) with probability at least \( 1 - \delta \), using at most
{\begin{align}
   N &=
   80 k\,\chi_h^2
   +
   \frac{463 \, n^2}{\varepsilon^2}\chi_t^2\\
   &=
   O\left(
   \bigl(n+\log\log\log(E)\bigr)\log\log(E)
   +
   \frac{n^3}{\varepsilon^2}
   \right)
\end{align}
measurement shots, where \( k \coloneqq \log \log \left( \|V^{-1}\|_{\infty} \right) \) is the number of times the adaptive unsqueezing is performed,
\begin{align}
\chi_h &\coloneqq \sqrt{2n}+\sqrt{2\log(4(k+1)/\delta)},\\
\chi_t &\coloneqq \sqrt{2n}+\sqrt{2\log(4/\delta)},
\end{align}
and \(E\) is the energy. In particular, in the natural regime \(n \gtrsim \log\log(E)\), the sample complexity reduces to
\bb
N=O\!\left(\frac{n^3}{\varepsilon^2}\right).
\ee}
\end{thm}

The idea behind this protocol is as follows. As discussed in the previous section, heterodyne tomography for non-squeezed states exhibits a sample complexity that is independent of the state's energy. This motivates an adaptive strategy: using only \( O(n) \) samples, the protocol first estimates the squeezing of the state and then applies to it an approximate “unsqueezing” Gaussian unitary. Each such operation exponentially increases the smallest eigenvalue of the covariance matrix. After \( k = O(\log\log(E)) \) such rounds---each requiring only a modest number of samples---the state becomes effectively non-squeezed. Standard heterodyne tomography can then be applied with sample complexity \( O(n^3 / \varepsilon^2) \).

Because the energy enters only through a double logarithm, even a state prepared using the energy of the entire universe would require at most a constant number of adaptive iterations (\( k \le 10 \)). Thus, the unsqueezing stage contributes negligibly to the overall sample cost, and the full tomography protocol achieves near energy-independence.

\begin{algorithm}[H]
\caption{Tomography of Gaussian states using adaptive unsqueezing}
\label{algo_ad}
\begin{algorithmic}[1]
\State \textbf{Input}: $kN_h+N_t$ copies of the Gaussian state $\rho$, failure probability $\delta$
\State $\hat S \gets \mathbb{1}$
\For{$i \in [k]$}
    \State {$\hat m,\hat V \gets \text{Heterodyne tomography}(N_h, U_{\hat S}\rho U^{\dagger}_{\hat{S}} ,\frac{\delta}{2(k+1)})$}
    \State $(\hat{D}_i, \hat{S}_i) \gets \mathrm{Williamson}(\hat{V})$
    \State $\hat S \gets \hat{S}_i^{-1} \hat{S}$
\EndFor
\State {$\hat m,\hat V \gets \text{Heterodyne tomography}(N_t, U_{\hat{S}}\rho U^{\dagger} _{\hat{S}},\frac{\delta}{2})$}
\State \textbf{Return:} ${\hat S}^{-1}\hat{m}$, $\hat S^{-1} \hat{V} \hat S^{-\top}$
\end{algorithmic}
{This algorithm estimates the Gaussian state using $kN_h + N_t$ copies of $\rho$, with success probability at least $1 - \delta$. Thanks to Theorem~14, setting
\bb
 k &\coloneqq \log \log \left( \|V^{-1}\|_{\infty} \right),\\
 N_h &\coloneqq 80\chi_h^2, \\
 N_t &\coloneqq \frac{463 \, n^2}{\varepsilon^2} \chi_t^2\,,
\ee
where
\bb
\chi_h &\coloneqq \sqrt{2n}+\sqrt{2\log(4(k+1)/\delta)},\\
\chi_t &\coloneqq \sqrt{2n}+\sqrt{2\log(4/\delta)},
\ee
guarantees a trace distance error of at most $\varepsilon$ with probability $\ge 1-\delta$.}

\end{algorithm}

\begin{figure}[t]
    \centering
    \includegraphics[width=\linewidth]{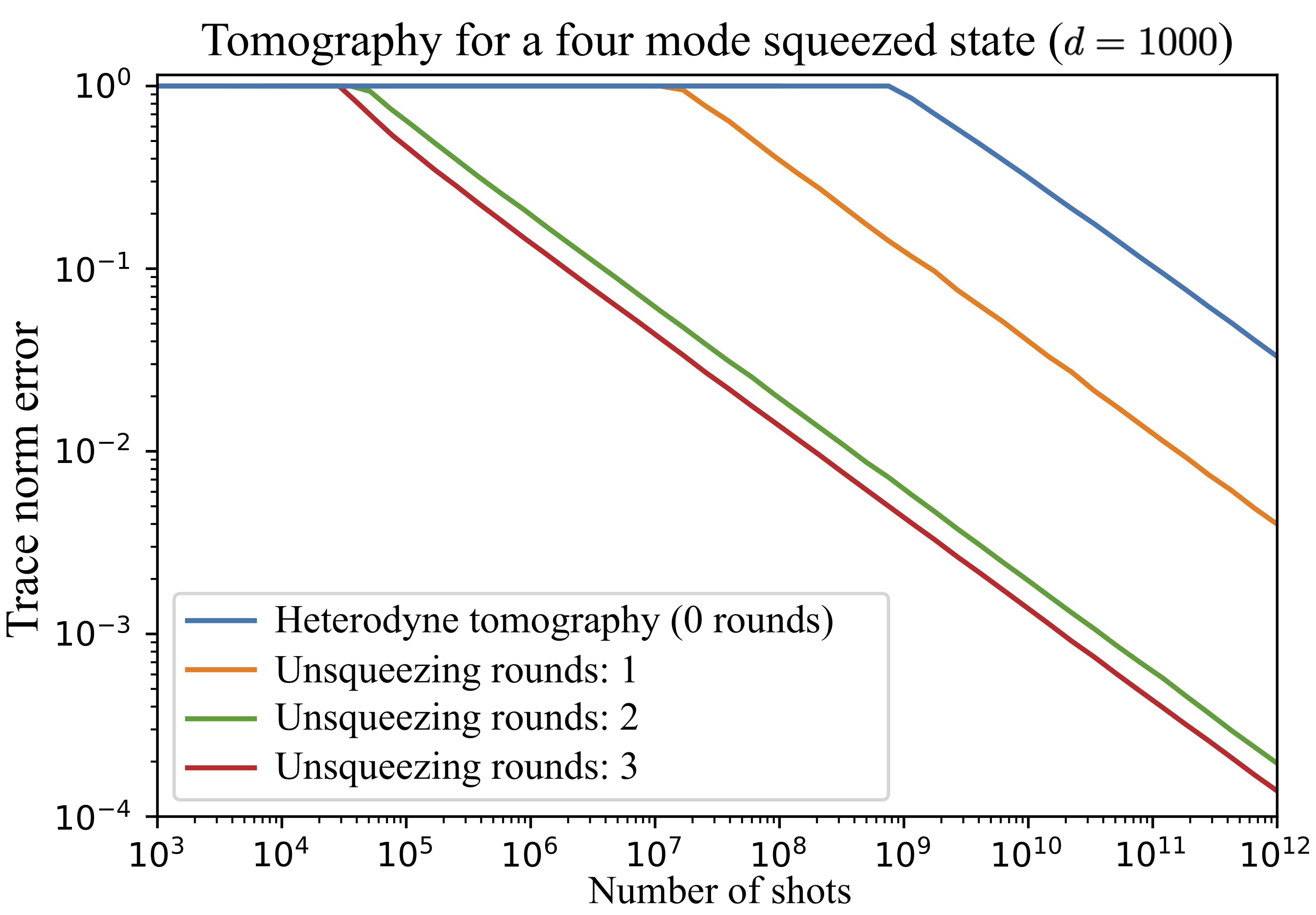}
    \caption{
    {Numerical illustration of the advantage of the adaptive unsqueezing strategy in Algorithm~\ref{algo_ad}. We consider a four-mode squeezed Gaussian state with covariance matrix \(V=\mathbb{1}_4\otimes\operatorname{diag}(d,1/d)\), with \(d=1000\), and compare standard heterodyne tomography with the adaptive protocol after one, two, and three unsqueezing rounds. The plot shows the trace-norm reconstruction error, evaluated using the trace-norm bound of Theorem~14, as a function of the number of measurement shots. Standard heterodyne tomography is strongly affected by the large squeezing and requires a very large number of samples before reaching the small-error regime. In contrast, even a few adaptive unsqueezing rounds substantially reduce this bottleneck, leading to much smaller sample sizes, in agreement with the nearly energy-independent scaling predicted by Theorem~\ref{th:adapTOMmain}.}
    }
    \label{fig:numerical_tomography_comparison}
\end{figure}

{Figure~\ref{fig:numerical_tomography_comparison} gives a numerical comparison between the adaptive Gaussian tomography protocol of Algorithm~\ref{algo_ad} and standard heterodyne tomography. 
The plot confirms the intuition that, for highly squeezed states, standard heterodyne tomography becomes inefficient because the heterodyne noise must be resolved against a strongly anisotropic covariance matrix. 
The adaptive protocol instead removes the squeezing of the unknown state step by step before the final heterodyne tomography stage, so that the remaining estimation problem is effectively that of a non-squeezed Gaussian state.
Consequently, in the high-squeezing, and hence high-energy, regime, the number of samples needed to reach a fixed trace-norm accuracy is reduced by several orders of magnitude, as also demonstrated numerically in Figure~\ref{fig:numerical_tomography_comparison}.}

\subsection*{Avoidance of active squeezing}

Algorithm~\ref{algo_ad} requires applying an active Gaussian unitary to the input state, followed by heterodyne detection. However, implementing active Gaussian operations is often experimentally demanding. Fortunately, this step can be entirely avoided: the combined effect of an active Gaussian transformation followed by heterodyne measurement can be exactly reproduced using only passive Gaussian operations, a fixed auxiliary squeezed input state, and homodyne measurements (plus classical post-processing)---all of which are routinely available with current photonic technologies~\cite{Schnabel2017}. In continuous-variable terminology, this alternative scheme corresponds to performing a specific \emph{generaldyne measurement} on the input state~\cite{BUCCO}.

This construction builds on the well-known equivalence between heterodyne detection and beam-splitter interference followed by homodyne measurements. As reviewed in the Preliminaries, heterodyne detection can be implemented by interfering the input state with vacuum on a balanced beam splitter, followed by homodyne detection of complementary quadratures---typically position on one output mode and momentum on the other. In the variant considered here, the vacuum input is replaced with a suitably chosen squeezed state. 
Key aspects of this measurement equivalence are illustrated in Fig.~1.3, with additional technical details provided in the supplementary material.

{Let us also comment on finite auxiliary squeezing. In the idealized analysis above, the auxiliary squeezed states are assumed to implement exactly the squeezing required by the adaptive unsqueezing step. In an experiment, however, the available offline squeezing may be bounded by some \(z_{\max}\). One can then replace each ideal unsqueezing step by a clipped version, where the squeezing parameters in the Bloch--Messiah decomposition are capped at \(z_{\max}\). In the clean setting where the symplectic part \(S\) in a Williamson decomposition of \(V\) is known, one can check that this clipped unsqueezing produces an effective covariance \(V_{\mathrm{eff}}\) satisfying \(\operatorname{Tr}(V_{\mathrm{eff}}^{-1})\leq 2n+\operatorname{Tr}V/z_{\max}^2\). Hence, if \(\operatorname{Tr}V\leq nE\), the final heterodyne step has sample complexity \(\widetilde O(n(n+nE/z_{\max}^2)^2/\varepsilon^2)\). Thus, finite auxiliary squeezing gives a smooth interpolation: the protocol remains essentially energy-independent for \(E\leq O(z_{\max}^2)\), while for larger \(E\) the residual scaling becomes \(\widetilde O(n^3E^2/(z_{\max}^4\varepsilon^2))\). A full end-to-end theorem with finite \(z_{\max}\) would require carrying through the same adaptive estimation and stability analysis as in the proof of Theorem~14, with clipping included at each iteration.}

\subsection*{Full energy-independence via access to the transpose}
While the adaptive tomography protocol presented in the previous section achieves sample complexity that is effectively energy-independent---scaling only double-logarithmically with the energy---it is natural to ask whether this residual dependence can be eliminated entirely.

This is indeed possible if one has access to the transposed state \( \rho^T\). In this setting, we design a tomography protocol whose sample complexity is fully independent of the state's energy, even in the presence of arbitrarily strong squeezing. The result is formalized in the following theorem; see Theorem~17 in the {supplemental material} for the proof and additional details.
\begin{thm}[{\normalfont (Energy-independent Gaussian state tomography with access to the transposed state; see Thm.~17)}]
\label{th:indeptomTR}
Let \( \varepsilon, \delta \in (0,1) \), and let \( \rho \) be an \( n \)-mode Gaussian state with covariance matrix \( V \).
{Then, by performing \(N\) rounds of a heterodyne-like measurement using \(\rho\) and \(\rho^T\) as inputs, the tomography algorithm described in Thm.~17 produces an estimated state \( \hat{\rho}=\rho(\hat V,\hat m) \) satisfying, with probability at least \( 1-\delta \),}
\begin{align}
   D_{\mathrm{tr}}(\rho, \hat{\rho}) \leq 8.7 n \frac{\chi}{\sqrt{N}},
\end{align}
where \( \chi = \sqrt{2n} + \sqrt{2 \log(2/\delta)} \) .
In particular, to achieve \( D_{\mathrm{tr}}(\rho, \hat{\rho}) \le \varepsilon \) with probability at least \( 1 - \delta \) it suffices to use
\begin{align}
   N = O\left( \frac{n^3}{\varepsilon^2} \right)
\end{align}
measurement shots.
\end{thm}
The measurement scheme consists of performing heterodyne detection on the input state \( \rho \), but instead of interfering it with vacuum (as in the standard implementation), the second port of the balanced beam splitter is fed with a copy of \( \rho^T \).

As a result, the scheme yields samples from a classical Gaussian distribution whose covariance matrix exactly matches that of the unknown state \( \rho \), thereby eliminating any dependence on the energy. In contrast, standard heterodyne detection introduces an additive identity shift---specifically, replacing \( V \) with \( (V + \mathbb{1})/2 \)---which is the source of the residual energy dependence.


\section*{Discussion}
In this work, we have introduced three conceptually distinct tomography algorithms for Gaussian quantum states, each with its own significance, 
revisiting the old problem of learning Gaussian quantum states from an entirely fresh perspective.
All algorithms are underpinned by a new trace-distance perturbation bound for Gaussian states, which serves as the technical foundation of our results.

First, we have equipped the standard heterodyne tomography protocol---routinely implemented in optical laboratories --- with rigorous trace-distance error guarantees. This result bridges the painful gap between experimental practice and the requirements of theoretical applications, where quantitative reconstruction guarantees in trace distance are often crucial. Our analysis offers simple and practical tools for experimentalists to assess the accuracy of their reconstructions. 
Second, we have presented an adaptive, sample-efficient, and experimentally feasible tomography algorithm whose sample complexity depends essentially only on the number of modes and is independent of the state’s energy up to doubly logarithmic factors. For all practical purposes, this yields an energy-independent protocol. This makes the algorithm particularly well-suited for quantum sensing applications~\cite{Schnabel_2017,Schnabel2010,SGravi2,q_sensing_cv,nguyen2024digitalreconstructionsqueezedlight,BUCCO}, where highly squeezed --- and therefore high-energy --- states are commonly used, especially in fundamental physics experiments.
Third, and finally, we have shown that even this mild residual energy dependence can be removed entirely: if access to the transposed state is available, a fully energy-independent tomography protocol becomes possible.

All three schemes are compatible with existing experimental capabilities in quantum optics laboratories around the world, relying only on tools already available in photonic platforms: auxiliary squeezed inputs, passive Gaussian unitaries reflecting passive optical elements, and homodyne detection.
Thus, we believe that these results will be of broad interest to both the experimental quantum optics community~\cite{Schnabel2017,Fusion,Xanadu,alexander2024manufacturable,Borealis,Zhong_2020,QuantumPhotoThermodynamics,_eh_ek_2015,teo2020highlyaccurategaussianprocess,Esposito_2014,Kumar_2022,kawasaki2024highrategenerationstatetomography,nguyen2024digitalreconstructionsqueezedlight,Schnabel2010}---where Gaussian states and their precise characterization play a central role---and the growing field of quantum learning theory~\cite{anshu2023survey,mele2024learningquantumstatescontinuous,zhao2023learning,huang2024learning,Huang_2020,grewal2023efficient,leone2023learning,hangleiter2024bell,bittel2024optimalestimatestracedistance,fanizza2023learning,mele2024efficient,iyer2025mildlyinteractingfermionicunitariesefficiently,austin2025efficientlylearningfermionicunitaries,arunachalam2023optimal,montanaro2017learning,rouzé2023learning,landau2024learningquantumstatesprepared,kim2024learningstatepreparationcircuits,bittel2024optimaltracedistanceboundsfreefermionic,gandhari_precision_2023,becker_classical_2023,oh2024entanglementenabled,fawzi2024optimalfidelityestimationbinary,möbus2023dissipationenabledbosonichamiltonianlearning,upreti2024efficientquantumstateverification,fanizza2024efficienthamiltonianstructuretrace,bittel2024optimalestimatestracedistance}, which aims to identify the optimal strategies for learning unknown quantum systems from measurements, with provable and precise recovery guarantees.

One can also view our work as the 
quantum analogue of a fundamental classical problem: learning Gaussian probability distributions in the total variation distance~\cite{Barsov1987,devroye2023total,arbas2023polynomial,franks2021nearoptimalsamplecomplexity,clement}, a task that plays a central role in many \emph{machine learning}  applications~\cite{arbas2023polynomial}. In the classical setting, the sample complexity of learning an $n$-dimensional Gaussian distribution within total variation distance $\varepsilon$ is $N = \Theta(n^2 / \varepsilon^2)$ (this follows by combining~\cite[Theorem~1.8]{arbas2023polynomial} and \cite[footnote~1]{clement}), which is also independent of the norm of the covariance matrix. The different scalings in \(n\) arise because the perturbation bound for total variation distance scales with the Frobenius norm of the covariance matrix, whereas in the quantum setting the corresponding bound scales with the trace norm of the covariance perturbation. 

Our work opens several natural directions for future research. Can similar energy-independent tomography protocols be developed for near-Gaussian or more general continuous-variable states? How robust are these schemes to experimental noise and imperfections? Can energy-independent methods be extended to process tomography of Gaussian unitaries? We leave these as promising open questions.

\section*{Methods}

\subsection*{Methods for the main results}

{Our analysis builds on the mathematical framework of Gaussian states in infinite-dimensional quantum systems, together with standard measurement techniques for continuous-variable systems. The necessary notation and preliminaries are introduced in Section~I of the Supplemental Material.

Section~II develops the technical tools required for our main stability analysis. In particular, it establishes Theorem~2, a perturbation bound for the trace distance between Gaussian states. The proof proceeds via two intermediate lemmas based on differentiability properties of Gaussian states and introduces techniques for estimating trace-distance perturbations from low-order moments. The resulting Theorem~8 provides a general perturbation bound for Gaussian states.

Section~III develops the tomography framework based on heterodyne measurements. We first derive sample-complexity bounds for estimating the first-moment vector and covariance matrix of multivariate Gaussian distributions and then apply these results to Gaussian state tomography. The resulting Theorem~10 establishes the sample complexity of a non-adaptive heterodyne tomography protocol.

Section~IV introduces an adaptive tomography protocol based on squeezed probe states. While Theorem~10 shows that the sample complexity of heterodyne tomography scales as $\|V^{-1}\|_\infty^2$ and therefore increases with the state's energy, the adaptive protocol effectively suppresses this dependence by iteratively reducing $\|V^{-1}\|_\infty$. The resulting Theorem~14 establishes a quasi-energy-independent tomography protocol.}

\subsection*{Further methods}

{Section~V presents a fully energy-independent tomography protocol under the additional assumption that the transpose $\rho^T$ of the unknown state $\rho$ is experimentally accessible. Building on the adaptive framework developed earlier, we show that this additional resource allows the sample complexity to become completely independent of the state's energy. The resulting Theorem~17 establishes a fully energy-independent Gaussian state tomography protocol.}

\subsection*{Acknowledgements}
The authors thank Clément Canonne, Senrui Chen, Marco Fanizza, Vittorio Giovannetti, Ludovico Lami, Lorenzo Leone, Salvatore F.~E.~Oliviero, Salvatore Tirone, and Nathan Walk for fruitful discussions. F.A.M.~thanks the California Institute of Technology for hospitality. 
This work has been supported by 
the BMFTR (PhoQuant, QPIC, Hybrid++, DAQC, QSolid), the BMWK (EniQmA), the DFG (CRC 183, {SPP 2514, BoLaCo}), the Quantum Flagship (PasQuans2, Millenion), the Munich Quantum Valley, Berlin Quantum, and the European Research Council (ERC AdG DebuQC).

\subsection*{Author contributions}

{All authors, L. Bittel, F.\ A.\ 
Mele, J.\ Eisert, and A.\ A.\ Mele have 
contributed to the research, the proofs, and the writing of the manuscript.}

\subsection*{Competing interests statement}

{The authors declare no competing financial or non-financial interests.}

\subsection*{Data availability statement}

{The core results of this work are of an analytical and rigorous nature.
The datasets generated and analyzed during the current study are not publicly available due to the way they are stored, but are available from the corresponding author on reasonable request.}

\subsection*{AI disclosure}
No generative AI was used to obtain the scientific results presented in this work.


\renewcommand{\refname}{References}
\bibliographystyle{apsrev4-2}
\bibliography{biblio}

\clearpage
\onecolumngrid
\setcounter{section}{0}
\renewcommand{\theequation}{S\arabic{equation}}
\setcounter{equation}{0}
\renewcommand{\theHequation}{SM.\arabic{equation}}
\renewcommand{\theHsection}{SM.\arabic{section}}
\setcounter{thm}{4}

\begin{center}
{\large\bfseries Supplemental material: Energy-independent tomography of Gaussian states\par}
\end{center}

\setcounter{tocdepth}{2} 
\tableofcontents
\vspace*{2\baselineskip}
 
In this supplemental material, we provide additional technical details and proofs to support the results presented in the main text. Specifically, we elaborate on the mathematical tools and concepts used in our analysis, outline the derivation of key bounds, and present the theoretical guarantees underpinning our Gaussian state tomography 
algorithm along with its algorithmic steps.
\resumetoc
\section{Notations and preliminaries}
In this section, we review the key concepts and preliminary notions that are essential for our work. For a more detailed yet concise discussion of these preliminaries, we direct the reader to the preliminary section in the appendix of our previous works~\cite{mele2024learningquantumstatescontinuous,bittel2024optimalestimatestracedistance}, where we adopt the same notation. For a comprehensive overview on notions of continuous variable systems and specifically on Gaussian states, 
we refer to the book~\cite{BUCCO} and the 
review articles~\cite{introeisert,weedbrook12}.
\smallskip\\

\textbf{Definitions}:
\begin{itemize}
    \item A $2n\times 2n$ real matrix $S$ is said to be \emph{symplectic} if $S\Omega S^\intercal=\Omega$, where $\Omega\coloneqq   \id_n \otimes \begin{pmatrix}
    0&1\\-1  &0
\end{pmatrix}$. Moreover, the set of all $2n\times 2n$ symplectic matrices is denoted as $\mathrm{Sp}(2n)$.
    \item An Hermitian matrix $X$ is said to be \emph{positive} if its minimum eigenvalue is non-negative. In formulae, we will write $X\geq0\Leftrightarrow \lambda_{\mathrm{min}}(X)\geq 0$. 
    \item An Hermitian matrix $X$ is said to be \emph{strictly positive} if its minimum eigenvalue is strictly 
    positive. In formulae, we will write $X>0\Leftrightarrow \lambda_{\mathrm{min}}(X)>0$. 
    \item The \emph{absolute value} of an operator $X$ is defined as \(|X| \coloneqq \sqrt{X^\dagger X}\).\\
    \item Given a vector \(\mathbf{m} \in \mathbb{R}^k\) and a positive definite 
    matrix $V\in \mathbb{R}^{k \times k}$, we denote as \(\mathcal{N}(\mathbf{m}, V):\mathbb{R}^k\longmapsto \mathbb{R}\) the \emph{Gaussian probability distribution} over \(\mathbb{R}^k\) with mean vector \(\mathbf{m} \) and covariance matrix \(V\), defined as
\begin{align}
    \mathcal{N}(\mathbf{m}, V)(\mathbf{r}) \coloneqq \frac{1}{(2\pi)^{k/2} \sqrt{\det V}} \exp\!\left( -\frac{1}{2} (\mathbf{r} - \mathbf{m})^\intercal V^{-1} (\mathbf{r} - \mathbf{m}) \right),
    \label{eq:classicalGaussian}
\end{align}
for all \(\mathbf{r} \in \mathbb{R}^k\).
\end{itemize}
\textbf{Decompositions}:
\begin{itemize}
    \item \emph{Spectral decomposition}: For every $X$ Hermitian, there exist a unitary $U$ and a real, diagonal matrix $D$ such that $X=U D U^{\dagger}$. The eigenvalues of $X$ are denoted $\lambda_i\coloneqq D_{i,i}$ for all $i$.
    \item \emph{Singular value decomposition}: For every $X$, there exist $U$ and $V$ unitaries and $\Sigma$ positive and diagonal such that $X=U \Sigma V^{\dagger}$. The singular values of $X$ are denoted as $\sigma_i\coloneqq\Sigma_{i,i}$ for all $i$.
    \item \emph{Williamson decomposition}: For every strictly positive matrix $X\in\mathbb{R}^{2n\times 2n}$, there exist a symplectic matrix $S$ and a positive diagonal matrix $D\in\mathbb{R}^{n\times n}$ such that $X=S\left(\mathbb{1}_2\otimes D\right)S^T$. The symplectic eigenvalues of $X$ are denoted as $d_{i,i}\coloneqq D_{i,i}$ for all $i$.
    \item \emph{Euler decomposition}~\cite{BUCCO}: For any real symplectic matrix $S$, 
    there exist  symplectic, orthogonal matrices $O_1$ and $O_2$, and a real diagonal matrix $Z \coloneqq \bigoplus_{j=1}^n \left( \begin{matrix} z_j & 0 \\ 0 & z_j^{-1} \end{matrix} \right)$ with each $z_j\ge 1$, such that  $S=O_1Z O_2$.
\end{itemize}
\textbf{Norms, positivity and inequalities}:
We use the following norms throughout the paper:
\begin{itemize}
\item $\|\mathbf{m}\|_2 \coloneqq \sqrt{\mathbf{m}^\dagger \mathbf{m}}$ denotes the \emph{Euclidean norm} of a vector $\textbf{m}$.
\item \(\|X\|_\infty\) denotes the \emph{operator norm}, and it is defined as the maximum singular value of \(X\).
\item \(\|X\|_1 \coloneqq \Tr|X|\) denotes the \emph{trace norm},
also referred to as \emph{one-norm}, 
and it is equal to the sum of singular values of $X$.
\item \(\|X\|_2 \coloneqq \sqrt{\Tr [X^\dagger X]}\) denotes the \emph{Hilbert-Schmidt} or \emph{Frobenius norm}, and it is equal to the Euclidean norm of the vector of singular values.
\item In general, for any $p>0$, the \emph{Schatten $p$-norm} is defined  as $\|X\|_p\coloneqq \left(\sum_i \sigma_i^p\right)^{1/p}$, where $\{\sigma_i\}_i$ are the singular values of $X$.
\end{itemize}
Useful matrix inequalities (see e.g., Ref.~\cite{BHATIA}) that are used in the proofs.
\begin{enumerate}
    \item \emph{Hölder's inequality}: $\|XY\|_p\leq \|X\|_q\|Y\|_r$, for $p,q,r\geq 1$ and  $p^{-1}\geq q^{-1}+r^{-1}$. 
    \item For $X$ Hermitian and $A\geq0$:  $\|\sqrt{A}X\sqrt{A}\|_1\leq\Tr(A|X|)\leq \|A|X|\|_1\leq\|AX\|_1$.
    \item For $A\ge B \ge 0 $, then $\|A\|_p\ge \|B\|_p$. 
    \label{eq:monotonicitynorm}
    \label{eq:listeqnorm1}
    \item For $A\geq B \geq 0$: $\|X A X^{\dagger}\|_1\geq \|X B X^{\dagger}\|_1$.
    \label{eq:listeqnorm2}
    \item For all matrices $A,B,C,D$: $\left\|\begin{pmatrix}
        A&B\\C&D
    \end{pmatrix}\right\|_1\geq \|A\|_1+\|D\|_1$\\
    \item For $X,Y\geq0$: $(X+Y)^{-1}\leq X^{-1}+Y^{-1}$.
    \item if $X\geq Y\geq 0$, then $X^{-1}\leq Y^{-1}$.
      \label{eq:inversebound}
\item If \( A \in \mathbb{C}^{n \times n} \) with \( n \ge 2 \) and \( A \ge \mathbb{1} \), then \( \|A\|_{\infty}^{1/2} \le \frac{1}{2} \|A\|_1 \).

    \label{eq:listeqtracebd}
\end{enumerate}
The \emph{matrix geometric mean} between two 
positive matrices $A$ and $B$ is defined as~\cite{BHATIA}    
\bb\label{eq:defgeommean}
    A\#B\coloneqq \sqrt{A}\sqrt{A^{-1/2}BA^{-1/2}}\sqrt{A}\,,
\ee
 and it satisfies the following properties~\cite{BHATIA}:
\begin{enumerate}
    \item If $[A, B] = 0$, then $A\#B = \sqrt{AB}$.
    \label{eq:commuting}
    \item  $A\#B = B\#A$.
    \item For $c \in \mathbb{R}$, $(cA)\#B = \sqrt{c}(A\#B)$.
    \label{eq:scalar}
    \item For invertible $X$, $(XAX^\dagger)\#(XBX^\dagger) = X(A\#B)X^\dagger$.
     \label{eq:congruence}
    \item Let $0\leq a\leq A$ and $0\leq b\leq B$ be four positive matrices. 
    Then $a\#b\leq A\#B$.
    \label{eq:geommeanmonoton}
\end{enumerate}
We refer to Refs.~\cite{Lami16,LL-log-det,revisited,bittel2024optimalestimatestracedistance} for other applications of the matrix geometric mean in quantum information with continuous-variable systems. 
\smallskip

\textbf{The Hilbert space of a continuous-variable system}: Continuous variable systems, such as quantum optical and bosonic systems, are associated with an infinite-dimensional Hilbert space, the \emph{$n$-mode bosonic Fock space} $\HH$. The latter can be defined as the span of all \emph{Fock states}:
\bb 
    \HH\coloneqq\mathrm{Span}\left\{\ket{m}: \, m\in\mathbb{N}_0^n\right\}\,,
\ee
where $\ket{m}$ is the $n$-mode Fock state vector with $m_i$ photons in the $i$th mode. More formally, $\HH$ can be defined to be the Hilbert space $L^2(\mathbb{R}^n)$ of all complex-values square-integrable over $\mathbb{R}^n$. The Fock states form an orthonormal basis for this Hilbert space. Quantum states on this space are referred to as \emph{$n$-mode} quantum states.

For each $i\in[n]$, the \emph{annihilation operator of the $i$th mode}, denoted as $a_i$, can be defined via its action on Fock states 
as
\bb
    a_i\ket{m}=\sqrt{m_i}\ket{m-e_i}\qquad\forall\, m\in\mathbb{N}_0^n\,,
\ee
where $e_i$ is the $i$th canonical basis vector of $\mathbb{R}^n$. Analogously, the \emph{creation operator of the $i$th mode}, denoted as $a_i^\dagger$, is the adjoint of the annihilation operator and satisfies
\bb
    a_i^\dagger\ket{m}=\sqrt{m_i+1}\ket{m+e_i}\qquad\forall\, m\in\mathbb{N}_0^n\,,
\ee
The annihilation and creation operator satisfy the  \emph{canonical commutation relations}
\bb
    [a_i,a_j^\dagger]=\delta_{i,j}\hat{\mathbb{1}}\,,
\ee
where $\delta_{i,j}$ denotes the Kronecker delta. Moreover, for each $j\in[n]$, the \emph{position operator} and the \emph{momentum operator} of the $j$th mode, denoted as $\hat{x}_j$ and $\hat{p}_j$, respectively, are defined as 
\bb
    \hat{x}_j&\coloneqq \frac{a_j+a_j^\dagger}{\sqrt{2}}\,,\\
    \hat{p}_j&\coloneqq \frac{a_j-a_j^\dagger}{\sqrt{2}i}\,.
\ee
In addition, the \emph{quadrature operator vector}, denoted as $\hat{\textbf{R}}$, is defined as
\bb
    \mathbf{\hat{R}} \coloneqq (\hat{x}_1, \hat{p}_1, \dots, \hat{x}_n, \hat{p}_n)^\intercal\,.
\ee
Using this notation, the canonical commutation relations can be expressed as 
\bb
    [\hat{R}_k,\hat{R}_l]=i\,(\Omega_n)_{k,l}\mathbb{\hat{1}}\qquad\forall\,k,l\in[2n]\,,
    \label{comm_rel_quadrature}
\ee
where 
\bb
    \Omega_n\coloneqq \bigoplus_{i=1}^n\left(\begin{matrix}0&1\\-1&0\end{matrix}\right)
\ee
being the so-called \emph{symplectic form}, and $\mathbb{\hat{1}}$ is the identity operator. The relation in~\eqref{comm_rel_quadrature} is usually expressed in the continuous variable literature~\cite{BUCCO} in vectorial notation as
\bb
[\hat{\textbf{R}},\hat{\textbf{R}}^{\intercal}]=i\,\Omega_n\mathbb{\hat{1}}\,.
\ee
The \emph{energy operator} is defined as
\bb
    \hat{E}_n \coloneqq \sum_{j=1}^n \left( \frac{\hat{x}_j^2}{2} + \frac{\hat{p}_j^2}{2} \right)=\frac{1}{2}\hat{\textbf{R}}^\intercal \hat{\textbf{R}}\,,
\ee
and the \emph{mean energy} of an $n$-mode state $\rho$ is given by its expectation value of the energy operator ${\mathrm \Tr}[\rho\hat{E}_n]$.

\textbf{First moments and covariance matrices}: The \emph{first moment} $m(\rho)$ of an $n$-mode quantum state $\rho$ is a $2n$-dimensional vector defined as 
\bb
    m(\rho)\coloneqq\left(\Tr\!\left[\hat{R}_1\,\rho\right],\Tr\!\left[\hat{R}_2\,\rho\right],\ldots,\Tr\!\left[\hat{R}_{2n}\,\rho\right]\right)\,,
\ee
which can be written in vectorial notation as $m(\rho)={\Tr} [{\hat{\textbf{R}}}\,\rho]$. 
Additionally, the \emph{covariance matrix} of $\rho$ is a $2n\times 2n$ matrix $V\!(\rho)$ with elements
\bb
	[V\!(\rho)]_{k,l}\coloneqq \Tr\!\left[\left\{{\hat{R}_k-m_k(\rho)\hat{\mathbb{1}},\hat{R}_l-m_l(\rho)\hat{\mathbb{1}}}\right\}\rho\right]= \Tr\!\left[\left\{\hat{R}_k,\hat{R}_l\right\}\rho\right]-2m_k(\rho)m_l(\rho) \, ,
\ee
for each $k,l\in[2n]$, where $\{\hat{A},\hat{B}\}\coloneqq \hat{A}\hat{B}+\hat{B}\hat{A}$ is the anti-commutator. In vectorial notation, this reads as
\bb
	V\!(\rho)&=\Tr\!\left[\left\{(\hat{\textbf{R}}-\textbf{m}(\rho)\hat{\mathbb{1}}),(\hat{\textbf{R}}-\textbf{m}(\rho)\hat{\mathbb{1}})^{\intercal}\right\}\rho\right]=  \Tr\!\left[\left\{{\hat{\textbf{R}}},{\hat{\textbf{R}}}^{\intercal}\right\}\rho\right]-2\textbf{m}(\rho)\textbf{m}(\rho)^\intercal \, .
\ee   
Notably, any covariance matrix $V\!(\rho)$ satisfies the matrix inequality~\cite{BUCCO}
\bb
V\!(\rho)+i\Omega_n\ge0\,,
\ee
known as \emph{uncertainty relation}. As a consequence, since $\Omega_n$ is skew-symmetric, any covariance matrix $V\!(\rho)$ is positive semi-definite on $\mathbb{R}^{2n}$. Conversely, for any symmetric $W\in\mathbb{R}^{2n,2n}$ such that $W+i\Omega_n\ge0$ there exists an $n$-mode state $\rho$ with covariance matrix $V\!(\rho)=W$~\cite{BUCCO}. In addition, the mean energy of an $n$-mode quantum state can be written in terms of its first moment and covariance matrix as
\bb\label{ene_cov_mat}
    \Tr[\hat{E}_n\rho]=\frac{\Tr V(\rho)}{4}+\frac{\|\textbf{m}\|_2^2}{2}\,.
\ee

\smallskip

\textbf{Symplectic matrices}: A $2n\times 2n$ real matrix $S$ is said to be \emph{symplectic} if $S\Omega S^\intercal=\Omega$, where $\Omega\coloneqq   \id_n \otimes \begin{pmatrix}
    0&1\\-1  &0
\end{pmatrix}$. Moreover, the set of all $2n\times 2n$ symplectic matrices is denoted as $S\in \mathrm{Sp}(2n)$.
\smallskip

\textbf{Williamson decomposition}: Any strictly positive $2n\times 2n$ real matrix \(V\) can be written in the so-called \emph{Williamson decomposition} as~\cite{BUCCO}
\begin{align}
    \label{eq:will}
    V = S D S^\intercal\,,
\end{align}
where $S$ is a symplectic matrix and
\begin{equation}
D = \bigoplus_{j=1}^n \left( \begin{matrix} d_j & 0 \\ 0 & d_j \end{matrix} \right),
\end{equation}
with \(d_1, d_2, \dots, d_n > 0\) denoting the \emph{symplectic eigenvalues}. If \(V\) is a valid covariance matrix, meaning it satisfies the uncertainty relation \(V + i\Omega \geq 0\), then all its symplectic eigenvalues obey \(d_1, d_2, \dots, d_n \geq 1\). 

 \smallskip

\textbf{Euler decomposition}: Any symplectic matrix \(S \in \mathrm{Sp}(2n)\) can be written in the Euler (or Bloch-Messiah) decomposition~\cite{BUCCO} as
\begin{align}
    \label{eq:Euler_dec}
    S = O_1 Z O_2,
\end{align}
where \(O_1, O_2 \in \mathrm{O}(2n) \cap \mathrm{Sp}(2n)\) are symplectic orthogonal matrices, and \(Z\) is a diagonal matrix given by
\begin{align}
    Z \coloneqq \bigoplus_{j=1}^n \left( \begin{matrix} z_j & 0 \\ 0 & z_j^{-1} \end{matrix} \right),
    \label{eq:Z}
\end{align}
with \(z_j \geq 1\). 
\smallskip

\textbf{Useful properties of covariance matrices}: Given the covariance matrix $V$ of an $n$-mode quantum state $\rho$, the Williamson decomposition and the Euler decomposition allow one to prove the following properties (see, e.g., Refs.~\cite{mele2024learningquantumstatescontinuous,bittel2024optimalestimatestracedistance}):
\begin{itemize}
    \item $\|V\|_{\infty} \geq 1$.
    \item $V^{-1}\le \Omega V \Omega^\intercal$, which directly implies that
    \bb
    \left(\lambda_{\mathrm{min}}(V)\right)^{-1}=\|V^{-1}\|_{\infty} \le \|V\|_{\infty}\le \Tr V\le 4\Tr[\hat{E}_n\rho]\,.
    \ee
    \end{itemize}
\smallskip

\textbf{Gaussian states}: By definition, Gaussian states are Gibbs states of positive quadratic Hamiltonians in the quadrature operator vector $\hat{\textbf{R}}$~\cite{BUCCO}. Moreover, they are uniquely determined by their covariance matrices and first moments. Specifically, the set of \(n\)-mode Gaussian states are in one-to-one correspondence with the set of pairs \((V, \mathbf{m})\), where \(V\) is a valid covariance matrix, i.e.,~a \(2n \times 2n\) real matrix that satisfies the uncertainty relation \(V + i\Omega \geq 0\), and \(\mathbf{m}\) is a \(2n\)-dimensional real vector~\cite{BUCCO}. For any such pair \((V, \mathbf{m})\), where \(V\) satisfies the uncertainty relation, there exists a unique Gaussian state with covariance matrix \(V\) and first moment \(\mathbf{m}\)~\cite{BUCCO}. Conversely, the covariance matrix of any (Gaussian) state must satisfy the uncertainty relation. We will denote the Gaussian state with covariance matrix \(V\) and first moment \(\mathbf{m}\) as \(\rho(V, \mathbf{m})\).

\smallskip

\textbf{Symplectic Gaussian unitaries (passive and active)}: 
Given a symplectic matrix \(S \in \mathrm{Sp}(2n)\), one can define the \(n\)-mode \emph{symplectic Gaussian unitary} \(U_S\)~\cite{BUCCO}, which transforms the quadrature operators as 
\begin{align}
    U_S^\dagger \hat{R}_k U_S = \sum_{k,l} S_{k,l} \hat{R}_l\qquad\forall\,k\in[2n]\,.
\end{align}
If a Gaussian unitary is associated with a symplectic orthogonal matrix, it is said to be \emph{passive}; otherwise, if the associated symplectic matrix is not orthogonal, it is said to be \emph{active}. Gaussian unitaries associated with symplectic matrices of the form \(Z\) in Eq.~\eqref{eq:Z} are active and are referred to as \emph{squeezing} unitaries. Importantly, passive unitaries preserve the energy operator~\cite{BUCCO}:
\bb
    U_O^\dagger \hat{E}_n U_O = \hat{E}_n\qquad\forall\, O \in \mathrm{O}(2n) \cap \mathrm{Sp}(2n)\,.
\ee
Finally, the Euler decomposition in Eq.~\eqref{eq:Euler_dec} implies that any Gaussian unitary can be written as a composition of passive and squeezing unitaries $U_S = U_{O_1} U_Z U_{O_2}$, where \(O_1, O_2\) are symplectic orthogonal matrices and \(Z\) is defined in Eq.~\eqref{eq:Z}.
\smallskip
 
\textbf{Homodyne measurements.}
A homodyne measurement~\cite{BUCCO,Weedbrook04} is a standard measurement in quantum optics that allows one to estimate the position and momentum of each bosonic mode. For an \(n\)-mode system with quadrature operator vector $\hat{\mathbf{R}} = (\hat{x}_1, \hat{p}_1, \dots, \hat{x}_n, \hat{p}_n)^\intercal,$
a homodyne measurement consists of selecting, for each mode \(j\), whether to measure \(\hat{x}_j\) or \(\hat{p}_j\).

If a Gaussian state \(\rho\), with first-moment vector \(\mathbf{m} \in \mathbb{R}^{2n}\) and covariance matrix \(V \in \mathbb{R}^{2n \times 2n}\), is measured via homodyne detection, the outcomes follow a Gaussian probability distribution whose mean and covariance are given by restricting \(\mathbf{m}\) and \(V/2\) to the measured quadratures. 
For example:
\begin{itemize}
    \item Measuring all position quadratures \((\hat{x}_1, \dots, \hat{x}_n)\) yields outcomes \(\mathbf{x} \in \mathbb{R}^n\) distributed as $\mathbf{x} \sim \mathcal{N}\left( \mathbf{m}_x, \, V_{x,x}/2 \right),$ where $\mathbf{m}_x  \coloneqq\left(m_{2j-1}(\rho )\right)_{j\in[n]}$ and \(V_{x,x}\coloneqq(V_{2i-1,2j-1})_{i,j\in[n]}\) is the submatrix of \(V\) corresponding to the position-position elements.

    \item Measuring all momentum quadratures \((\hat{p}_1, \dots, \hat{p}_n)\) gives outcomes $\mathbf{p} \sim \mathcal{N}\left( \mathbf{m}_p, \, V_{p,p}/2 \right),$
    with $\mathbf{m}_p \coloneqq\left(m_{2j}(\rho )\right)_{j\in[n]}$ and  $V_{p,p}\coloneqq(V_{2i,2j})_{i,j\in[n]}$ is the submatrix of \(V\) corresponding the momentum-momentum elements.
\end{itemize}
More generally, fixed a symplectic matrix $S$, the homodyne measurement associated with $S$ is a projective measurement with respect to the common (generalised) eigenbasis of the $n$ commuting observables $\{\sum_{j=1}^{2n} S_{i,j}\hat{\textbf{R}}_j\}_{i\in\{1,3,5,\ldots,2n-1\}}$. Measuring all these $n $ observables on an $n$-mode Gaussian state $\rho$ yields outcomes \(\mathbf{x} \in \mathbb{R}^n\) distributed as $\mathbf{x} \sim \mathcal{N}\left( \mathbf{m}_x, \, V_{x,x}/2 \right),$ where \(\mathbf{m}_x \coloneqq \left(m_{2j-1}(U_S\rho U_s^\dagger)\right)_{j\in[n]}$ and \(V_{x,x}\coloneqq(V(U_S\rho U_s^\dagger)_{2i-1,2j-1})_{i,j\in[n]}\). For example, the homodyne detection associated with the symplectic orthogonal matrix 
\begin{eqnarray}
S=\left(\begin{array}{cc}
    \cos\theta& \sin\theta\\
    -\sin\theta& \cos\theta
\end{array}
\right)
\end{eqnarray}
corresponds to measuring  the quadrature observable \(X_\theta \coloneqq  \hat{x} \cos\theta + \hat{p} \sin\theta\).

Experimentally, homodyne measurements are implemented by interfering the signal mode with a strong, phase-controlled \emph{local oscillator} (LO) at a beam splitter, followed by intensity measurements at the output ports~\cite{BUCCO}. The phase of the LO determines the measured quadrature: a phase of \(0\) corresponds to \(\hat{x}\), a phase of \(\pi/2\) to \(\hat{p}\), and other values to rotated quadratures.

\vspace{1em}
\textbf{Heterodyne measurement.}
The heterodyne measurement~\cite{BUCCO,Weedbrook04} allows one to simultaneously extract (noisy) information about both the position and momentum quadratures of each bosonic mode. 
When applied to a Gaussian state \(\rho\), with first-moment vector \(\mathbf{m} \in \mathbb{R}^{2n}\) and covariance matrix \(V \in \mathbb{R}^{2n \times 2n}\), the heterodyne outcomes follow a classical Gaussian distribution
\begin{align}
\mathbf{r} \sim \mathcal{N}\left( \mathbf{m}, \, \frac{V + \id}{2} \right),
\end{align}
where \(\id\) denotes the \(2n \times 2n\) identity matrix.

Experimentally, heterodyne detection is implemented by evolving the system state \(\rho\) together with an auxiliary vacuum state \(\ket{0}\) through a balanced beam splitter \(U_B\), followed by homodyne detection on both output ports: one measures the position quadratures on one arm and the momentum quadratures on the other.
The beam splitter \(U_B\) is a passive Gaussian operation corresponding to a symplectic orthogonal transformation
\begin{align}
B = H \otimes \id_{2n}, \quad \text{where} \quad H \coloneqq \frac{1}{\sqrt{2}} \begin{pmatrix} 1 & 1 \\ 1 & -1 \end{pmatrix}.
\end{align}
More generally, we have the following.
\begin{lemma}[(Generalized heterodyne sampling via beam splitter and homodyne detection)]
\label{le:generalized_het}
Let \(\rho(V_1, m_1)\) and \(\rho(V_2, m_2)\) be two \(n\)-mode Gaussian states with first-moment vectors \(m_1, m_2 \in \mathbb{R}^{2n}\) and covariance matrices \(V_1, V_2 \in \mathbb{R}^{2n \times 2n}\). Consider the product state \(\rho(V_1, m_1) \otimes \rho(V_2, m_2)\), and evolve it under the beam splitter transformation \(U_B\), corresponding to the passive Gaussian unitary associated with the symplectic orthogonal matrix \(B = H \otimes \id_{2n}\).  
Then, perform homodyne measurements of the position quadratures \(\hat{x}_1, \dots, \hat{x}_n\) on the first output arm and the momentum quadratures \(\hat{p}_{n+1}, \dots, \hat{p}_{2n}\) on the second. The resulting classical measurement outcome \(\mathbf{r} = \sqrt{2} (\boldsymbol{x}_1, \boldsymbol{p}_2) \in \mathbb{R}^{2n}\) is distributed according to
\begin{align}
\mathbf{r} \sim \mathcal{N}\left( m_1 + \bar{m}_2, \, \frac{V_1 + \bar{V}_2}{2} \right),
\end{align}
where
\begin{align}
\bar{m}_2 \coloneqq (\mathbb{1}_n \otimes Z) m_2, \qquad \bar{V}_2 \coloneqq (\mathbb{1}_n \otimes Z)V_2 (\mathbb{1}_n \otimes Z),
\end{align}
with \(Z \coloneqq \mathrm{diag}(1, -1)\).
In particular, if \(\rho(V_2, m_2)\) is the vacuum state (i.e., \(m_2 = 0\) and \(V_2 = \id_{2n}\)), then we recover the heterodyne measurement result
\begin{align}
\mathbf{r} \sim \mathcal{N} \left( m_1, \, \frac{V_1 + \id_{2n}}{2} \right).
\end{align}
\end{lemma}

\begin{proof}
The joint state has first moment and covariance matrix
\begin{align}
m &= \begin{pmatrix} m_1 \\ m_2 \end{pmatrix}, \quad
V = \begin{pmatrix} V_1 & 0 \\ 0 & V_2 \end{pmatrix}.
\end{align}
After applying the beam splitter, the state remains Gaussian with transformed first moment and covariance
\begin{align}
m' &= B m = \frac{1}{\sqrt{2}} \begin{pmatrix} m_1 + m_2 \\ m_1 - m_2 \end{pmatrix}, \\
V' &= B V B^\intercal = \frac{1}{2} \begin{pmatrix} V_1 + V_2 & V_1 - V_2 \\ V_1 - V_2 & V_1 + V_2 \end{pmatrix}.
\end{align}
We then measure \(\coloneqq\hat{x}_1, \dots, \hat{x}_n\) on the first output arm and \(\coloneqq \hat{p}_{n+1}, \dots, \hat{p}_{2n}\) on the second through homodyne detection. The outcomes \((\boldsymbol{x}_1, \boldsymbol{p}_2)\) are Gaussian-distributed with mean and covariance obtained by restricting \(m'\) and \(V'/2\) to the measured quadratures. In particular, the restricted output mean vector and covariance matrix are given by
\begin{align}
m'\big|_{(\boldsymbol{x}_1, \boldsymbol{p}_2)} 
&= \frac{1}{\sqrt{2}} \begin{pmatrix}
(m_1)|_{\boldsymbol{x}} + (m_2)|_{\boldsymbol{x}} \\
(m_1)|_{\boldsymbol{p}} - (m_2)|_{\boldsymbol{p}}
\end{pmatrix}, \\
\frac{1}{2}V'\big|_{(\boldsymbol{x}_1, \boldsymbol{p}_2)} 
&= \frac{1}{4} \begin{pmatrix}
(V_1)|_{\boldsymbol{x}, \boldsymbol{x}} + (V_2)|_{\boldsymbol{x}, \boldsymbol{x}} 
& (V_1)|_{\boldsymbol{x}, \boldsymbol{p}} - (V_2)|_{\boldsymbol{x}, \boldsymbol{p}} \\
(V_1)|_{\boldsymbol{p}, \boldsymbol{x}} - (V_2)|_{\boldsymbol{p}, \boldsymbol{x}} 
& (V_1)|_{\boldsymbol{p}, \boldsymbol{p}} + (V_2)|_{\boldsymbol{p}, \boldsymbol{p}}
\end{pmatrix}.
\end{align}
Thus, \((\boldsymbol{x}_1, \boldsymbol{p}_2) \sim \mathcal{N}\left( \frac{m_1 + \bar{m}_2}{\sqrt{2}}, \, \frac{V_1 + \bar{V}_2}{4} \right) \) from which we have \(\sqrt{2}(\boldsymbol{x}_1, \boldsymbol{p}_2) \sim \mathcal{N}\left( m_1 + \bar{m}_2, \, \frac{V_1 + \bar{V}_2}{2} \right) \).

\end{proof}



\section{Trace distance bounds for Gaussian states}
This section is dedicated to proving a perturbation bound for the trace distance between Gaussian states, presented in Theorem~2. This bound plays a crucial role in analyzing the performance of our tomography algorithm. Additionally, it is a result of independent interest, with potential applications in continuous variable quantum information theory~\cite{BUCCO}. 
%
%
The proof of Theorem~2 is deferred to the end of this section. To establish this result, we introduce two lemmas that provide key intermediate results. These rely on the differentiability properties of Gaussian states, as developed in Ref.~\cite{bittel2024optimalestimatestracedistance}.

\begin{lemma}[(Perturbation bound: first moment)]
Let \(\rho(D, \alpha t')\) be a Gaussian state with diagonal covariance matrix \(D = \mathbb{1}_2 \otimes d\), where \(d \geq \mathbb{1}\), and first-moment vector \(\alpha t'\), with \(\alpha \in \mathbb{R}\). Then,
\begin{align}
    \left\| \left. \frac{\partial}{\partial \alpha} \rho(D, \alpha t') \right|_{\alpha=0} \right\|_1 
    = \left\| \sum_{k=1}^{2n} t_k \left[ R_k, \rho(D, 0) \right] \right\|_1 
    \leq 2 \left\| (D + \mathbb{1})^{-1/2} t' \right\|_2,
\end{align}
where we denoted with $t\coloneqq \Omega   t'$.
\end{lemma}

\begin{proof}

The first identity in the claim follows directly from Theorem S13 (Compact formula for the derivative of a Gaussian state) in Ref.~\cite{bittel2024optimalestimatestracedistance}.
We now prove the upper bound. Consider first the single-mode case, where the Gaussian state with diagonal covariance matrix, denoted as \(\rho_\tau\), is a thermal state with parameter \(\tau = \frac{d - 1}{d + 1}\) (and equivalently \(d = \frac{\tau + 1}{1 - \tau}\)); see, e.g., Ref.~\cite{BUCCO}. In particular, the thermal state takes the form

\begin{align}
    \rho_\tau = (1 - \tau) \sum_{m=0}^\infty \tau^m \ketbra{m}{m}.
\end{align}
We compute the commutator of the position operator \(x = \frac{1}{\sqrt{2}}(a + a^\dagger)\) with \(\rho_\tau\) as
    \begin{align}
        [x,\rho_\tau]&=\frac{(1-\tau)}{\sqrt{2}} \sum_{m=0}^\infty \tau^m \left( (a+a^\dagger )\ketbra{m}{m}-\ketbra{m}{m}(a+a^\dagger )\right)\\ 
        \nonumber
        &=\frac{(1-\tau)}{\sqrt{2}}\sum_{m=0}^\infty \tau^m \left( -\sqrt{m+1}\ketbra{m}{m+1}+\sqrt{m+1}\ketbra{m+1}{m}\right)\times (1-\tau)\\
        \nonumber
        &=\frac{(1-\tau)^2}{\sqrt{2}}\sum_{m=0}^\infty \tau^m\sqrt{m+1}\left( \ketbra{m+1}{m}-\ketbra{m}{m+1}\right).
        \nonumber
    \end{align}
    This generalizes to the multi mode case as 
    \begin{align}
        \sum_i [t_i R_i,\rho_D]&=\frac{1}{\sqrt{2}}\sum_{m=0}^\infty p_{\vec m}\sum_{i\in [n]}(t_{2i-1}-it_{2i})(1-\tau_i)  \left( \sqrt{m_i+1}\ketbra{\vec m+\vec e_i}{\vec m}\right) -h.c.,
    \end{align}
    where $p_{\vec m}=\prod_{i=1}^n (1-\tau_i) \tau_i^{m_i}$ is the probability distribution of the Fock state elements. 
   %
    Thus, taking the one-norm, we get
    \begin{align}
        \left\|\sum_{i\in[2n]} [t_i R_i,\tau_\nu]\right\|_1&\leq \frac{1}{\sqrt{2}}\sum_{\vec m=0}^\infty p_{\vec m}\left\|\sum_{i\in [n]}(t_{2i-1}-it_{2i})(1-\tau_i) \left( \sqrt{m_i+1}\ketbra{\vec m+\vec e_i}{\vec m}\right) -h.c.\right\|_1\\
        \nonumber
        &\leq \frac{2}{\sqrt{2}}\sum_{\vec m=0}^\infty p_{\vec m}\left\|\sum_{i\in [n]}(t_{2i-1}-it_{2i})(1-\tau_i)  \sqrt{m_i+1}\ketbra{\vec m+\vec e_i}{\vec m}\right\|_1\\
        \nonumber
        &=\frac{2}{\sqrt{2}}\sum_{\vec m=0}^\infty p_{\vec m}\left\|\sum_{i\in [n]}(t_{2i-1}-it_{2i})(1-\tau_i)   \sqrt{m_i+1}\ket{\vec m+\vec e_i}\right\|_2\\
        \nonumber
        &\leq \frac{2}{\sqrt{2}}\sqrt{\sum_{\vec m=0}^\infty p_{\vec m}\left\|\sum_{i\in [n]}(t_{2i-1}-it_{2i})(1-\tau_i)   \sqrt{m_i+1}\ket{\vec m+\vec e_i}\right\|^2_2}\\
        \nonumber
        &=\frac{2}{\sqrt{2}}\sqrt{\sum_{\vec m=0}^\infty p_{\vec m}\sum_{i\in [n]}(t_{2i-1}^2+t_{2i}^2)(1-\tau_i)^2   (m_i+1)}
        \nonumber
    \end{align}
     where we have used that $\mathbb{E}(X)\leq \sqrt{\mathbb{E}(X^2)}$. 
    Next, we can use that
    \begin{align}
        (1-p)\sum_{m=0}^\infty p^m (m+1)=\frac{1}{1-p},
    \end{align}
    to obtain 
    \begin{align}
        \left\|\sum_{i\in[2n]} [t_i R_i,\tau_\nu]\right\|_1&\leq \frac{2}{\sqrt{2}}\sqrt{\sum_{i\in [n]}(t_{2i-1}^2 +t_{2i}^2)(1-\tau_i) }=2\| (D+\mathbb{1})^{-1/2} t\|_2=2\| (D+\mathbb{1})^{-1/2} t'\|_2,
    \end{align}
    using $1-\tau_i=\frac{2}{d_i+1}$ and that $\Omega$ commute with $D$, which conclude the proof.
\end{proof}
\begin{lemma}
Let $\rho(D, 0)$ be a bosonic Gaussian state with diagonal covariance matrix $D = \mathbb{1}_2 \otimes d$, where $d \geq \mathbb{1}$. Then,
\begin{align}
    \left\| \left. \frac{\partial}{\partial \alpha} \rho(D + \alpha X, 0) \right|_{\alpha = 0} \right\|_1 
    = \frac{1}{4} \left\| \sum_{k, \ell = 1}^{2n} (\Omega X \Omega^\top)_{k , \ell} \left[ R_k, \left[ R_\ell, \rho(D, 0) \right] \right] \right\|_1 
    \leq \frac{1 + \sqrt{3}}{2} \left\| (D + \mathbb{1})^{-1/2} X (D + \mathbb{1})^{-1/2} \right\|_1.
\end{align}
\end{lemma}
\begin{proof}
Let $\tilde{X} \coloneqq \Omega X \Omega^\top$. By Theorem~S13 (Compact formula for the derivative of a Gaussian state) in Ref.~\cite{bittel2024optimalestimatestracedistance}, we have
\begin{align}
    \left. \frac{\partial}{\partial \alpha} \rho(D + \alpha X, 0) \right|_{\alpha = 0}
    = -\frac{1}{4} \sum_{k, \ell = 1}^{2n} \tilde{X}_{k , \ell} \left[ R_k, \left[ R_\ell, \rho \right] \right].
\end{align}
We now expand the quadrature operators \( R_k \) in terms of creation and annihilation operators
\begin{align}
    R_{2j - 1 + \alpha} = \frac{1}{\sqrt{2}} i^\alpha \left( (-1)^\alpha a_j + a_j^\dagger \right), \qquad \alpha \in \{0,1\},\; j \in [n].
\end{align}
Substituting this into the double commutator, we find
\begin{align}
    -4\, \partial_{X,0} \rho(D+ \alpha X,0)
    &= \sum_{i,j \in [2n]} \tilde{X}_{i,j} \left[ \hat{R}_i, \left[ \hat{R}_j, \rho \right] \right] \\
    \nonumber
    &= \frac{1}{2} \sum_{i,j \in [n]} \sum_{\alpha, \beta \in \{0,1\}} i^{\alpha + \beta} \tilde{X}_{2i - 1 + \alpha,\, 2j - 1 + \beta} 
    \left[ (-1)^\alpha a_i + a_i^\dagger,\; \left[ (-1)^\beta a_j + a_j^\dagger, \rho \right] \right] \\
     \nonumber
    &= \frac{1}{2} \sum_{i,j \in [n]} \sum_{\alpha, \beta \in \{0,1\}} i^{\alpha + \beta} \tilde{X}_{2i - 1 + \alpha,\, 2j - 1 + \beta} \Big(
    [a_i^\dagger, [a_j^\dagger, \rho]] + (-1)^\alpha [a_i, [a_j^\dagger, \rho]] \notag\\
    &\qquad\qquad\qquad\qquad\qquad + (-1)^\beta [a_i^\dagger, [a_j, \rho]] + (-1)^{\alpha + \beta} [a_i, [a_j, \rho]]
    \Big) \\
     \nonumber
    &= \sum_{i,j \in [n]} \left( 
    A_{i,j} [a_i^\dagger, [a_j^\dagger, \rho]] 
    +  B_{i,j} [a_i, [a_j^\dagger, \rho]]
    +  B^{*}_{{i,j}} [a_i^\dagger, [a_j, \rho]] 
    + A_{{i,j}}^{*} [a_i, [a_j, \rho]] \right) \\
    &= \sum_{i,j \in [n]} \left( 
    A_{{i,j}} [a_i^\dagger, [a_j^\dagger, \rho]] 
    + 2 B_{{i,j}} [a_i, [a_j^\dagger, \rho]] 
    + A_{{i,j}}^{*} [a_i, [a_j, \rho]]
    \right),
     \nonumber
\end{align}
where in the third line we have defined the matrices
\begin{align}
    A_{{i,j}} &= \frac{1}{2} \sum_{\alpha, \beta \in \{0,1\}} i^{\alpha + \beta} \tilde{X}_{2i - 1 + \alpha,\, 2j - 1 + \beta}, \\
    B_{{i,j}} &= \frac{1}{2} \sum_{\alpha, \beta \in \{0,1\}} i^{\alpha + \beta} (-1)^\alpha \tilde{X}_{2i - 1 + \alpha,\, 2j - 1 + \beta},
\end{align}
that satisfy \( A = A^\top \), \( B = B^\dagger \). 
We now note that
\begin{align}
    \sum_{i,j \in [n]} B^{*}_{{i,j}} [a^\dagger_i, [a_j, \rho]] 
    &= \sum_{i,j \in [n]} B_{j,i} [a^\dagger_i, [a_j, \rho]] \\
    \nonumber
    &= - \sum_{i,j \in [n]} B_{j,i} \left( [\rho, [a^\dagger_i, a_j]] + [a_j, [\rho, a^\dagger_i]] \right) \\
    \nonumber
    &= - \sum_{i,j \in [n]} B_{j,i} [a_j, [\rho, a^\dagger_i]] \\
    \nonumber
    &= \sum_{i,j \in [n]} B_{j,i} [a_j, [a^\dagger_i, \rho]]\\
    \nonumber
    &= \sum_{i,j \in [n]} B_{{i,j}} [a_i, [a^\dagger_j, \rho]],
    \nonumber
\end{align}
where in the first step we have used Hermiticity of \( B \), in the second step we applied the Jacobi identity
\begin{equation}
[A, [B, C]] = -[B, [C, A]] - [C, [A, B]],
\end{equation}
and in the third step we have used the canonical commutation relation \( [a_i, a^\dagger_j] = \delta_{{i,j}} \).
Thus, we have
\begin{align}
\label{eq:derpr}
    -4\, \partial_{X,0} \rho(V,0)
    &= \sum_{i,j \in [n]} \left( 
    A_{{i,j}} [a_i^\dagger, [a_j^\dagger, \rho]] 
    + 2 B_{{i,j}} [a_i, [a_j^\dagger, \rho]] 
    + A_{{i,j}}^{*} [a_i, [a_j, \rho]]
    \right).
\end{align}
Note that the matrices \( A \) and \( B \) can be obtained from the matrix \( \tilde{X} = \Omega X \Omega^\top \) via the unitary transformation
\begin{align}
    M = \frac{1}{\sqrt{2}} \begin{pmatrix}
        \mathbb{1} & \mathbb{1} \\
        i\mathbb{1} & -i\mathbb{1}
    \end{pmatrix}, \qquad 
    M^\dagger \tilde{X} M = 
    \begin{pmatrix}
        B & A \\
        A^\dagger & B^*
    \end{pmatrix}.
\end{align}
 We also introduce the notation
\begin{align}
    A_\tau \coloneqq D_\tau A D_\tau, \qquad 
    B_\tau \coloneqq D_\tau B D_\tau, \qquad 
    X_\tau \coloneqq (D_\tau \oplus D_\tau)\, X\, (D_\tau \oplus D_\tau),
\end{align}
with
\begin{align}
    D_\tau \coloneqq \operatorname{diag}\left( \sqrt{1 - \tau_1}, \dots, \sqrt{1 - \tau_n} \right).
\end{align}
We start by estimating the contribution from the \(B\)-term in Eq.~\eqref{eq:derpr}. As in the proof of the previous lemma, we note that we can write the state \(\rho\) as a diagonal thermal state in the Fock basis as
\begin{align}
\rho = \sum_{\vec{m} \in \mathbb{N}^n} p_{\vec{m}} \ketbra{\vec{m}}{\vec{m}}, \qquad 
p_{\vec{m}} \coloneqq \prod_{i=1}^n (1 - \tau_i) \tau_i^{m_i},
\end{align}
with \(\tau_i = \frac{d_i - 1}{d_i + 1}\). We have
  \begin{align}
      \sum_{i,j \in [n]} B_{{i,j}} \left[ a_i, \left[ a_j^\dagger, \rho \right] \right]
    &= \sum_{i,j \in [n]} B_{{i,j}} \sum_{\vec{m} \in \mathbb{N}^n} p_{\vec{m}} \left(
        a_i a_j^\dagger \ketbra{\vec{m}}{\vec{m}} 
        + \ketbra{\vec{m}}{\vec{m}} a_j^\dagger a_i 
        - a_i \ketbra{\vec{m}}{\vec{m}} a_j^\dagger 
        - a_j^\dagger \ketbra{\vec{m}}{\vec{m}} a_i 
    \right) \\
    \nonumber
    &= \sum_{\vec{m} \in \mathbb{N}^n} p_{\vec{m}} \sum_{i \in [n]} B_{i,i} \left(
        (2 m_i + 1) \ketbra{\vec{m}}{\vec{m}} 
        - (m_i + 1) \ketbra{\vec{m} + \vec{e}_i}{\vec{m} + \vec{e}_i} 
        - m_i \ketbra{\vec{m} - \vec{e}_i}{\vec{m} - \vec{e}_i}
    \right) \\
    \nonumber
    &\quad + \sum_{\vec{m} \in \mathbb{N}^n} p_{\vec{m}} \sum_{\substack{i,j \in [n] \\ i \neq j}} B_{{i,j}} \Big(
        \sqrt{m_j + 1} \sqrt{m_i} \ketbra{\vec{m} + \vec{e}_j - \vec{e}_i}{\vec{m}} 
        + \sqrt{m_j} \sqrt{m_i + 1} \ketbra{\vec{m}}{\vec{m} - \vec{e}_j + \vec{e}_i} \nonumber \\
        \nonumber
    &\hspace{5em}
        - \sqrt{m_j + 1} \sqrt{m_i + 1} \ketbra{\vec{m} + \vec{e}_j}{\vec{m} + \vec{e}_i} 
        - \sqrt{m_j} \sqrt{m_i} \ketbra{\vec{m} - \vec{e}_i}{\vec{m} - \vec{e}_j}
    \Big).
    \nonumber
\end{align}
We now simplify the previous expression by appropriately shifting the summation indices. For this, we use the relation
\begin{align}
p_{\vec{m}} \coloneqq \prod_{i=1}^n (1 - \tau_i) \tau_i^{m_i}, \qquad 
\Rightarrow \quad
p_{\vec{m} + \vec{e}_i} = \tau_i p_{\vec{m}}, \qquad 
{p_{\vec{m} - \vec{e}_i}} = \frac{{p_{\vec{m}}}}{\tau_i}.
\end{align}
 So we get
\begin{align}
    \sum_{i,j \in [n]} B_{{i,j}} \left[ a_i, \left[ a_j^\dagger, \rho \right] \right]
    &= \sum_{\vec{m} \in \mathbb{N}^n} p_{\vec{m}} 
       \sum_{i \in [n]} B_{i,i} \left( 
           2m_i + 1 - \tau_i(m_i + 1) - \frac{m_i}{\tau_i} 
       \right) \ketbra{\vec{m}}{\vec{m}} \notag \\
       \nonumber
    &\quad - \sum_{\vec{m} \in \mathbb{N}^n} p_{\vec{m}} 
       \sum_{\substack{i, j \in [n] \\ i \neq j}} 
       B_{{i,j}} \sqrt{(m_i + 1)(m_j + 1)} 
       (1 - \tau_i)(1 - \tau_j) 
       \ketbra{\vec{m} + \vec{e}_j}{\vec{m} + \vec{e}_i} \\ 
    &= \sum_{\vec{m} \in \mathbb{N}^n} p_{\vec{m}} 
       \sum_{i \in [n]} B_{i,i} \left( 
           2m_i + 1 - \tau_i(m_i + 1) - \frac{m_i}{\tau_i} 
           + \frac{m_i(1 - \tau_i)^2}{\tau_i} 
       \right) \ketbra{\vec{m}}{\vec{m}} \notag \\
       \nonumber
    &\quad - \sum_{\vec{m} \in \mathbb{N}^n} p_{\vec{m}} 
       \sum_{i,j \in [n]} B_{{i,j}} \sqrt{(m_i + 1)(m_j + 1)} 
       (1 - \tau_i)(1 - \tau_j) 
       \ketbra{\vec{m} + \vec{e}_j}{\vec{m} + \vec{e}_i} \\ 
    &= \sum_{\vec{m} \in \mathbb{N}^n} p_{\vec{m}} 
       \sum_{i \in [n]} B_{i,i} (1 - \tau_i) \ketbra{\vec{m}}{\vec{m}} \notag \\
    &\quad - \sum_{\vec{m} \in \mathbb{N}^n} p_{\vec{m}} 
       \sum_{i,j \in [n]} B_{{i,j}} \sqrt{(m_i + 1)(m_j + 1)} 
       (1 - \tau_i)(1 - \tau_j) 
       \ketbra{\vec{m} + \vec{e}_j}{\vec{m} + \vec{e}_i}.
       \nonumber
\end{align}
    We will proceed with the bound
    \begin{align}
        \|\sum_{i,j\in[n]} B_{i,j}[ a_i[a_j^\dagger,\rho]]-\sum_{i}B_{i,i}(1-\tau_i)\rho \|_1&\leq
        \sum_{\vec m\in \mathbb{N}^n}p_{\vec m}\left\|\sum_{i, j\in [n]}B_{i,j}(\sqrt{m_j+1}\sqrt{m_i+1}(1-\tau_i)(1-\tau_j)\ketbra{j}{i}\right\|_1 .
    \end{align}
We write the matrix \( B_{\tau} \) in its eigendecomposition \( B_{\tau} = \sum_k \lambda_k \ket{b_k} \bra{b_k} \), yielding
    \begin{align}
    &\sum_{\vec m\in \mathbb{N}^n}p_{\vec m}\left\|\sum_{i, j\in [n]}B_{{i,j}}(\sqrt{m_j+1}\sqrt{m_i+1}(1-\tau_i)(1-\tau_j)\ketbra{j}{i}\right\|_1\\
    \nonumber
    &\leq \sum_k |\lambda_k|\sum_{\vec m\in \mathbb{N}^n}p_{\vec m}\left\|\sum_{j\in [n]}b_{kj}^{*}(\sqrt{m_j+1}(1-\tau_j)^{1/2}\ket{j}\times \sum_{i\in [n]} \bra{i}\sqrt{m_i+1}(1-\tau_i)^{1/2}b_{k,i}\right\|_1
    \\
     \nonumber
    &\leq \sum_k |\lambda_k|\sum_{\vec m\in \mathbb{N}^n}p_{\vec m}\left\|\sum_{j\in [n]}b_{kj}^{*}(\sqrt{m_j+1}(1-\tau_j)^{1/2}\ket{j}\right\|_2^2\\
     \nonumber
        &=\sum_k |\lambda_k| \sum_{\vec m\in \mathbb{N}^n}p_{\vec m}\sum_i (1-\tau_i)(m_i+1)|b_{k,i}|^2\\
        &=\sum_k |\lambda_k| \sum_i |b_{k,i}|^2\\
         \nonumber
        &=\Tr(|B_\tau|)=2\|(D+\mathbb{1})^{-1/2}B (D+\mathbb{1})^{-1/2}\|_1
         \nonumber
    \end{align}
where we have used that for rank-one operators, 
\(\| \ketbra{a}{b} \|_1 = \| \ket{a} \|_2 \| \ket{b} \|_2\), and that \(p_{\vec{m}}\) is a probability distribution. Finally, we have used the identity
\begin{align}
    (1 - \tau_i)^2 \sum_{m = 0}^\infty \tau_i^m (m + 1) =1,
\end{align} 
denoted \(\tau \coloneqq \mathrm{Diag}(\vec{\tau})\), \(D \coloneqq \mathrm{Diag}(\vec{d})\), and used that $1-\tau_i=\frac{2}{d_i+1}$.
We now consider the term involving the matrix \(A\). Specifically,
\begin{align}
    \sum_{i,j \in [n]} A_{{i,j}} [a_i^{\dagger}, [a_j^\dagger, \rho]] 
    &= \sum_{i,j \in [n]} A_{{i,j}} \sum_{\vec{m} \in \mathbb{N}^n} p_{\vec{m}} 
    \Big( 
        a_i^\dagger a_j^\dagger \ketbra{\vec{m}}{\vec{m}} 
        + \ketbra{\vec{m}}{\vec{m}} a_j^\dagger a_i^\dagger 
        - a_i^\dagger \ketbra{\vec{m}}{\vec{m}} a_j^\dagger 
        - a_j^\dagger \ketbra{\vec{m}}{\vec{m}} a_i^\dagger 
    \Big) \\
     \nonumber
    &= \sum_{\vec{m} \in \mathbb{N}^n} p_{\vec{m}} \Big(
        \sqrt{m_i + 1 + \delta_{{i,j}}} \sqrt{m_j + 1} 
        \ketbra{\vec{m} + \vec{e}_i + \vec{e}_j}{\vec{m}} \notag \\
         \nonumber
    &\quad + \sqrt{m_i - \delta_{{i,j}}} \sqrt{m_j} 
        \ketbra{\vec{m}}{\vec{m} - \vec{e}_i - \vec{e}_j} \notag \\
         \nonumber
    &\quad - \sqrt{m_i + 1} \sqrt{m_j} 
        \ketbra{\vec{m} + \vec{e}_i}{\vec{m} - \vec{e}_j} 
        - \sqrt{m_i} \sqrt{m_j + 1} 
        \ketbra{\vec{m} + \vec{e}_j}{\vec{m} - \vec{e}_i} 
    \Big) \\
     \nonumber
    &= \sum_{i,j \in [n]} A_{{i,j}} \sum_{\vec{m} \in \mathbb{N}^n} p_{\vec{m}} 
    \sqrt{m_j + 1} \sqrt{m_i + 1 + \delta_{{i,j}}} 
    (1 - \tau_i)(1 - \tau_j) 
    \ketbra{\vec{m} + \vec{e}_i + \vec{e}_j}{\vec{m}}.
     \nonumber
\end{align}
This shows that the combined contribution from the \(A\)-term and the diagonal part of the \(B\)-term is bounded by
\begin{align}
    &\left\| \sum_{i,j \in [n]} A_{{i,j}} [a_i^{\dagger}, [a_j^\dagger, \rho]] 
    + \sum_{i} B_{i,i} (1 - \tau_i) \rho \right\|_1 
    \\
    \nonumber
    &\leq \sum_{\vec{m} \in \mathbb{N}^n} p_{\vec{m}} 
    \sqrt{ \sum_{i,j} |A_{{i,j}}|^2 (m_j + 1)(m_i + 1 + \delta_{{i,j}}) 
    (1 - \tau_i)^2 (1 - \tau_j)^2 
    + \left( \sum_i B_{i,i}(1 - \tau_i) \right)^2 } \\ 
    \nonumber
    &\leq \sqrt{ \sum_{\vec{m} \in \mathbb{N}^n} p_{\vec{m}} 
    \sum_{i,j} |A_{{i,j}}|^2 (m_j + 1)(m_i + 1 + \delta_{{i,j}}) 
    (1 - \tau_i)^2 (1 - \tau_j)^2 
    + \left( \sum_i B_{i,i}(1 - \tau_i) \right)^2 } \\ 
    \nonumber
    &= \left( 
    \sum_{i \neq j} |A_{{i,j}}|^2 
    \left( \sum_{m = 0}^\infty \tau_i^m (m + 1)(1 - \tau_i)^3 \right)
    \left( \sum_{m = 0}^\infty \tau_j^m (m + 1)(1 - \tau_j)^3 \right) \right. \notag \\ 
    \nonumber
    &\quad + \left. 
    \sum_i |A_{i,i}|^2 (1 - \tau_i)^5 
    \sum_{m = 0}^\infty \tau_i^m (m + 1)(m + 2) 
    + \left( \sum_i B_{i,i} (1 - \tau_i) \right)^2 
    \right)^{1/2} \\
    \nonumber
    &= \sqrt{ 
    \sum_{i \neq j} |A_{{i,j}}|^2 (1 - \tau_i)(1 - \tau_j) 
    + 2 \sum_i |A_{i,i}|^2 (1 - \tau_i)^2 
    + \left( \sum_i B_{i,i} (1 - \tau_i) \right)^2 } \\
    \nonumber
    &= \sqrt{ 
    \sum_{i,j} |A_{{i,j}}|^2 (1 - \tau_i)(1 - \tau_j) 
    + \sum_i |A_{i,i}|^2 (1 - \tau_i)^2 
    + \left( \sum_i B_{i,i} (1 - \tau_i) \right)^2 }.
    \nonumber
\end{align}
As in the earlier part of the proof, we applied the triangle inequality, along with 
\(\| \ketbra{a}{b} \|_1 = \| \ket{a} \|_2 \| \ket{b} \|_2\). 
We also used the bound \(\mathbb{E}(X) \leq \sqrt{\mathbb{E}(X^2)}\), which follows from Jensen's inequality.
In addition, we made use of the identity
\begin{align}
    \sum_{m = 0}^{\infty} p^m (m + 1)(m + 2) = \frac{2}{(1 - p)^3}, \qquad \text{for } 0 < p < 1.
\end{align}
Putting everything together, we obtain
\begin{align}
    &\left\| \left. \partial_\alpha \rho(D + \alpha X, 0) \right|_{\alpha = 0} \right\|_1 
    \\
    \nonumber
    &\leq \frac{1}{2} \sqrt{ \sum_{i,j} |A_{{i,j}}|^2 (1 - \tau_i)(1 - \tau_j) 
    + \sum_i |A_{i,i}|^2 (1 - \tau_i)^2 
    + \left( \sum_i (1 - \tau_i) B_{i,i} \right)^2 } 
    + \frac{1}{2} \sum_i (1 - \tau_i) |B_{i,i}| \\
     \nonumber
    &\leq \frac{1}{2} \sqrt{ 2 \sum_{i,j} |A_{{i,j}}|^2 (1 - \tau_i)(1 - \tau_j) 
    + \left( \sum_i (1 - \tau_i) B_{i,i} \right)^2 } 
    + \frac{1}{2} \sum_i (1 - \tau_i) |B_{i,i}| \\
     \nonumber
    &\leq \frac{1}{2} \sqrt{ 2 \| A_\tau \|_2^2 + \Tr(B_\tau)^2 } 
    + \frac{1}{2} \| B_\tau \|_1 \\
    \nonumber
    &\leq \frac{1}{2} \sqrt{ 2 \| A_\tau \|_1^2 + \| B_\tau \|_1^2 } 
    + \frac{1}{2} \| B_\tau \|_1 \\
     \nonumber
    &\leq \frac{1}{2} \sqrt{ \frac{3}{4} \| \tilde{X}_\tau \|_1^2 } 
    + \frac{1}{4} \| \tilde{X}_\tau \|_1 \\
     \nonumber
    &= \frac{1 + \sqrt{3}}{4} \| \tilde{X}_\tau \|_1 \\
    \nonumber
    &= \frac{1 + \sqrt{3}}{2} \left\| (D + \mathbb{1})^{-1/2} \tilde{X} (D + \mathbb{1})^{-1/2} \right\|_1 \\
     \nonumber
    &= \frac{1 + \sqrt{3}}{2} \left\| (D + \mathbb{1})^{-1/2} X (D + \mathbb{1})^{-1/2} \right\|_1,
     \nonumber
\end{align}
where the steps are justified as follows.
we have used that \( X \) admits the block representation
\begin{align}
    \tilde{X} = M 
    \begin{pmatrix}
        B & A \\
        A^\dagger & B^*
    \end{pmatrix}
    M^\dagger,
\end{align}
under the unitary transformation
\begin{align}
    M = \frac{1}{\sqrt{2}} \begin{pmatrix}
        \mathbb{1} & \mathbb{1} \\
        i\mathbb{1} & -i\mathbb{1}
    \end{pmatrix}.
\end{align}
Since \( D_\tau \oplus D_\tau \) commutes with \( M \), the rescaled matrix \( \tilde{X}_\tau \) satisfies
\begin{align}
    \tilde{X}_\tau = M 
    \begin{pmatrix}
        B_\tau & A_\tau \\
        A_\tau^\dagger & B_\tau^*
    \end{pmatrix}
    M^\dagger.
\end{align}
From this, we observe that
\begin{align}
    \| \tilde{X}_\tau \|_1 
    = \left\| 
    \begin{pmatrix}
        B_\tau & A_\tau \\
        A_\tau^\dagger & B_\tau^*
    \end{pmatrix}
    \right\|_1 
    \geq \max\left\{ \left\| 
    \begin{pmatrix}
        A_\tau & 0 \\
        0 & A_\tau^\dagger
    \end{pmatrix}
    \right\|_1, 
    \left\| 
    \begin{pmatrix}
        B_\tau & 0 \\
        0 & B_\tau^*
    \end{pmatrix}
    \right\|_1 
    \right\}
    = 2 \max\left\{ \| A_\tau \|_1, \| B_\tau \|_1 \right\},
\end{align}
as follows from Statement 6 in the list of preliminary inequalities. This justifies the use of \( \| \tilde{X}_\tau \|_1 \) to upper bound the trace norm contributions from both \( A_\tau \) and \( B_\tau \).
Finally, we have used the identity \( 1 - \tau_i = \frac{2}{d_i + 1} \), and in the last step, the fact that \( D \) commutes with \( \Omega \).

\end{proof}
We are now finally in the position to prove the trace-distance bound, by assembling the results established above.

\begin{thm}[(Perturbation bound for Gaussian states)]
\label{th:pertboundAPP}
Let $\rho(V,t)$ be a bosonic Gaussian state. Then, we have
\begin{align}
    \left\| \left. \frac{\mathrm{d}}{\mathrm{d} \alpha} \rho(V + \alpha X, m + \alpha t) \right|_{\alpha = 0} \right\|_1
    \leq \frac{1 + \sqrt{3}}{2} \Tr\left( V^{-1} |X| \right) + 2 \| V^{-1/2} t \|_2.
\end{align}
Moreover, for any two Gaussian states $\rho(V, t)$ and $\rho(W, m)$, the bound 
\begin{align}
    \| \rho(V, t) - \rho(W, m) \|_1 
    &\leq 2\| V^{-1/2} (m - t) \|_2 
    + \frac{1 + \sqrt{3}}{4} \Tr\left( (V^{-1} + W^{-1}) |V - W| \right).
\end{align}
\end{thm}
\begin{proof}
We use the Williamson decomposition \( V = S D S^T \) to obtain
\begin{align}
    \rho(V + \alpha X, m + \alpha t)
    &= U_m \rho(V + \alpha X, \alpha t) U_m^\dagger \\
    \nonumber
    &= U_m\rho(S ( D + \alpha S^{-1} X S^{-T}) S^T, \alpha t)U^{\dagger}_m \\
    \nonumber
    &= U_m U_S \, \rho(D + \alpha S^{-1} X S^{-T},\, \alpha S^{-1} t) \, U_S^\dagger U_m^\dagger,
     \nonumber
\end{align}
where we have used that both displacement of the first moment and squeezing correspond to unitary transformations, which in particular leave the trace norm invariant. Therefore, computing 
\begin{align}
    \left\| \left. \partial_\alpha \rho(V + \alpha X, m + \alpha t) \right|_{\alpha = 0} \right\|_1
\end{align}
is equivalent to computing
\begin{align}
    \left\| \left. \frac{\partial}{\partial \alpha} \rho(D + \alpha S^{-1} X S^{-T},\, \alpha S^{-1} t) \right|_{\alpha = 0} \right\|_1.
\end{align}
We then estimate this trace norm derivative by splitting it as
\begin{align}
    &\left\| \left. \frac{\partial}{\partial \alpha} \rho(D + \alpha S^{-1} X S^{-T},\, \alpha S^{-1} t) \right|_{\alpha = 0} \right\|_1  \\
     \nonumber
    &\leq 
    \left\| \left. \frac{\partial}{\partial \alpha} \rho(D + \alpha S^{-1} X S^{-T},\, 0) \right|_{\alpha = 0} \right\|_1 
    + \left\| \left. \frac{\partial}{\partial \alpha} \rho(D,\, \alpha S^{-1} t) \right|_{\alpha = 0} \right\|_1 \\
     \nonumber
    &\leq 
    \frac{1 + \sqrt{3}}{2} \left\| (D + \mathbb{1})^{-1/2} S^{-1} X S^{-T} (D + \mathbb{1})^{-1/2} \right\|_1 
    + 2 \left\| (D + \mathbb{1})^{-1/2} S^{-1} t \right\|_2 \\
     \nonumber
    &\leq 
    \frac{1 + \sqrt{3}}{2} \left\| D^{-1/2} S^{-1} X S^{-T} D^{-1/2} \right\|_1 
    + 2 \left\| D^{-1/2} S^{-1} t \right\|_2,
     \nonumber
\end{align}
where in the first step we have used the definition of derivative and the triangle inequality, in the second step we have applied the previously established lemmas, and in the last step we have removed the identity component for simplicity. In the special case of pure states (\(D = \mathbb{1}\)), this last step can be avoided. As such, a more refined analysis could yield sharper sampling bounds in the pure-state regime.
Now observe that
\begin{align}
    V = S D S^T = (S D^{1/2})(D^{1/2} S^T),
\end{align}
so that
\begin{align}
    V^{-1/2} = U D^{-1/2} S^{-1}
\end{align}
for some unitary matrix \( U \). Since Schatten norms are unitary-invariant, we get
\begin{align}
    \left\| \left. \frac{\mathrm{d}}{\mathrm{d} \alpha} \rho(V + \alpha X, m + \alpha t) \right|_{\alpha = 0} \right\|_1 
    &\leq \frac{1 + \sqrt{3}}{2} \left\| V^{-1/2} X (V^{-1/2})^{T} \right\|_1 
    + 2 \left\| V^{-1/2} t \right\|_2 \\
     \nonumber
    &= \frac{1 + \sqrt{3}}{2} \left\| V^{-1/2} X V^{-1/2} \right\|_1 
    + 2 \left\| V^{-1/2} t \right\|_2,
\end{align}
which concludes the first part of the proof.
Next,
\begin{align}
    \|\rho(V, t) - \rho(W, m)\|_1 
    &\leq \|\rho(V, t) - \rho(V, m)\|_1 
    + \|\rho(V, m) - \rho(W, m)\|_1  \\
     \nonumber
    &\leq \int_0^1 \left\| \partial_\alpha \rho(V, t + \alpha(m - t)) \right\|_1\, \mathrm{d}\alpha 
    + \int_0^1 \left\| \partial_\alpha \rho(V + \alpha(W - V), m) \right\|_1\, \mathrm{d}\alpha  \\
     \nonumber
    &\leq 2 \int_0^1 \|V^{-1/2}  (m - t)\|_2\, \mathrm{d}\alpha 
    + \frac{1 + \sqrt{3}}{2} \int_0^1 
    \Big\| (V + \alpha(W - V))^{-1/2} \notag \\
     \nonumber
    &\hspace{3.5cm}
    \cdot (W - V)\, (V + \alpha(W - V))^{-1/2} \Big\|_1\, \mathrm{d}\alpha \\
    &=2\| V^{-1/2} (m - t) \|_2 
    + \frac{1 + \sqrt{3}}{4} \Tr\left( (V^{-1} + W^{-1}) |V - W| \right),
     \nonumber
\end{align}
where we have used
\begin{align}
    &\int_0^1 \left\| (V + \alpha(W - V))^{-1/2} (W - V) (V + \alpha(W - V))^{-1/2} \right\|_1 \mathrm{d}\alpha \\
     \nonumber
    &\leq \int_0^1 \Tr\left[ (V + \alpha(W - V))^{-1} |W - V| \right] \mathrm{d}\alpha \\
    \nonumber
    &\leq \int_0^1 \Tr\left[ \big((1 - \alpha)V^{-1} + \alpha W^{-1}\big) |W - V| \right] \mathrm{d}\alpha \\
     \nonumber
    &= \Tr\left( \frac{V^{-1} + W^{-1}}{2} |W - V| \right).
     \nonumber
\end{align}

\end{proof}

\section{Gaussian state tomography using heterodyne measurements: A non-adaptive algorithm}
We begin by deriving bounds on the number of samples required to estimate the first moment vector and the covariance matrix of an unknown multivariate classical Gaussian distribution.
These bounds will later be applied to the analysis of heterodyne measurement outcomes.
\begin{lemma}[(Relative error estimation of the covariance matrix)]
\label{le:estrelative}
Let $\{\hat x_i\}_{i \in [N]}$ be $N$ i.i.d.\ samples from an $n$-dimensional Gaussian distribution $\mathcal{N}(\mu, \Sigma)$ (see Eq.~\eqref{eq:classicalGaussian}). Define the empirical mean and covariance estimators as
\begin{align}
    \hat \mu &\coloneqq \frac{1}{N} \sum_{i \in [N]} \hat x_i, \qquad
    \hat \Sigma \coloneqq \frac{1}{N} \sum_{i \in [N]} (\hat x_i - \hat \mu)(\hat x_i - \hat \mu)^T = \frac{1}{N} \sum_{i \in [N]} \hat x_i \hat x_i^T - \hat \mu \hat \mu^T.
\end{align}
Then, with probability at least $1 - \delta$, the bounds 
\begin{align}
\label{eq:firstmom}
    \|\mu - \hat \mu\|_2 &\leq \frac{\chi}{\sqrt{N}} \sqrt{\|\Sigma\|_\infty}, \\
    \|\hat \Sigma - \Sigma\|_\infty &\leq \zeta \|\Sigma\|_\infty,
    \label{eq:opnorm}
\end{align}
hold, where $\chi \coloneqq \sqrt{n} + \sqrt{2 \log(2/\delta)}$ and $\zeta\coloneqq\frac{2\chi}{\sqrt{N}} + \frac{2\chi^2}{N}$. Eq.\eqref{eq:firstmom} and Eq.\ \eqref{eq:opnorm} also imply that
\begin{align}
        \|\Sigma^{-\frac{1}{2}}(\hat{\mu}_N - \mu)\|_2  \leq \frac{\chi}{\sqrt{N}},  \\
    (1 - \zeta)\Sigma \leq \hat{\Sigma} \leq (1 + \zeta)\Sigma.
    \label{eq:relative}
\end{align}
\end{lemma}

\begin{proof}
The bound on the mean estimator follows from Eq.~(1.1) in Ref.~\cite{lugosi2017subgaussianestimatorsmeanrandom}, using the inequality $\|\Sigma\|_1 \leq n \|\Sigma\|_\infty$.
For the covariance bound, consider
\begin{align}
    \|\hat \Sigma - \Sigma\|_\infty &= \left\| \frac{1}{N} \sum_{i \in [N]} (\hat x_i - \mu)(\hat x_i - \mu)^T - (\hat \mu - \mu)(\hat \mu - \mu)^T - \Sigma \right\|_\infty
    \nonumber
    \\
     \nonumber
    &\leq \left\| \frac{1}{N} \sum_{i \in [N]} (\hat x_i - \mu)(\hat x_i - \mu)^T - \Sigma \right\|_\infty + \|(\hat \mu - \mu)(\hat \mu - \mu)^T\|_\infty 
    \\
     \nonumber
    &\leq \left( \frac{2\chi}{\sqrt{N}} + \frac{\chi^2}{N} \right) \|\Sigma\|_\infty + \|\hat \mu - \mu\|_2^2 \\
     \nonumber
    &\leq \left( \frac{2\chi}{\sqrt{N}} + \frac{\chi^2}{N} \right) \|\Sigma\|_\infty + \frac{\chi^2}{N} \|\Sigma\|_\infty \\
     \nonumber
    &= \left( \frac{2\chi}{\sqrt{N}} + \frac{2\chi^2}{N} \right) \|\Sigma\|_\infty \\
    &= \zeta \|\Sigma\|_\infty,
    \label{eq:proofeq}
\end{align}
where the inequality in the third line uses Eq.~(6.12) from Ref.~\cite{Wainwright_2019}.
To prove Eq.~\eqref{eq:relative}, note that this relation is equivalent to requiring that
\begin{align}
    - \zeta \mathbb{1} \leq \Sigma^{-\frac{1}{2}} \hat{\Sigma} \Sigma^{-\frac{1}{2}} - \mathbb{1} \leq \zeta \mathbb{1},
\end{align}
which is true if and only if 
\begin{align}
    \|\Sigma^{-\frac{1}{2}} \hat{\Sigma} \Sigma^{-\frac{1}{2}} - \mathbb{1}\|_{\infty} \le \zeta.
\end{align}
Now observe that $\Sigma^{-1/2}\hat{\Sigma} \Sigma^{-1/2}$ can be written as
\begin{align}
    \Sigma^{-1/2}\hat{\Sigma}\Sigma^{-1/2} = \frac{1}{N} \sum_{i=1}^N (\xi_i - \hat{\xi})(\xi_i - \hat{\xi})^\intercal,
\end{align}
where $\xi_i \coloneqq \Sigma^{-1/2}(\hat{x}_i - \mu)$ is distributed according to $\mathcal{N}(0, \mathbb{1})$~\cite{vershynin_2018}, and $\hat{\xi} \coloneqq \frac{1}{N} \sum_{j=1}^N \xi_j$ is the empirical mean.
Therefore, $\Sigma^{-1/2} \hat{\Sigma} \Sigma^{-1/2}$ is the empirical covariance matrix for samples drawn from $\mathcal{N}(0, \mathbb{1})$, and hence, by Eq.\eqref{eq:proofeq}, the bound $\|\Sigma^{-\frac{1}{2}} \hat{\Sigma} \Sigma^{-\frac{1}{2}} - \mathbb{1}\|_{\infty} \le \zeta$ follows. Similarly, we have that by Eq.\eqref{eq:firstmom}, $$\|\Sigma^{-\frac{1}{2}}(\hat{\mu}_N - \mu)\|_2 = \|\hat{\xi} - 0 \|_2 \leq \frac{\chi}{\sqrt{N}}. $$
\end{proof}
This allows us to state our first main theorem concerning the recovery guarantees of tomography of an unknown Gaussian state via heterodyne measurements. The key observation behind this theorem is that the probability distribution resulting from a heterodyne measurement on a Gaussian state $\rho$ is itself Gaussian~\cite{BUCCO}, with mean vector $m(\rho)$ and covariance matrix $(V(\rho) + \id)/2$, where $V(\rho)$ and $m(\rho)$ denote the covariance matrix and mean vector of $\rho$, respectively. Hence, by performing heterodyne measurements on the unknown Gaussian state, we can use the previous lemma to obtain a relative error estimate on the classical covariance matrix $(V(\rho) + \id)/2$. Accepting this, the remainder of the proof is dedicated to showing how such a relative error propagates to a bound on the trace distance between the true state and the reconstructed Gaussian state.

We show that performing only heterodyne measurements, without any squeezing, results in a tomography algorithm whose sample complexity scales at most quadratically with \(\|V^{-1}\|_{\infty}\) (i.e., at most quadratically with the energy). This improves upon previous bounds for tomography using solely heterodyne measurements~\cite{bittel2024optimalestimatestracedistance}. 

\begin{thm}[(Non-adaptive Gaussian state tomography via heterodyne measurements)]
\label{th:heterotom}
    Let $\{\hat r_i\}_{i \in [N]}$ be $N$ samples obtained from a heterodyne measurement of a Gaussian state $\rho(V, m)$. Define the estimators
    \begin{align}
        \hat m &\coloneqq  \hat{\mu}_h, \quad 
        \hat V \coloneqq \frac{2 \hat \Sigma_h}{1 - \zeta} - \mathbb{1},
    \end{align}
    where $\hat \mu_h \coloneqq \frac{1}{N} \sum_{i \in [N]} \hat r_i$ and $\hat \Sigma_h \coloneqq \frac{1}{N} \sum_{i \in [N]} (\hat r_i - \hat \mu_h)(\hat r_i - \hat \mu_h)^T$ are the empirical estimators of the heterodyne samples, and
    \begin{align}
    \label{eq:zetadef}
        \zeta &\coloneqq \left( \frac{2 \chi}{\sqrt{N}} + \frac{2 \chi^2}{N} \right), \quad
        \chi \coloneqq \sqrt{2n} + \sqrt{2 \log \left( \frac{2}{\delta} \right)}.
    \end{align}
    Then, with probability at least $1 - \delta$, we have
\begin{align}
        \frac{1}{2}\|\rho(\hat V, \hat m) - \rho(V, m)\|_1 &\leq 4.3  \left(2n + \Tr(V^{-1})\right)\frac{\sqrt{2n} + \sqrt{2 \log \left( 2 \delta^{-1} \right)}}{\sqrt{N}} .
    \end{align}
This also implies that, for any $\varepsilon \in (0,1)$, choosing  
\begin{align}
   N > \left( \frac{4.3}{\varepsilon} \left( 2n+ \Tr(V^{-1})\right) \left(\sqrt{2n} + \sqrt{2 \log \left( 2\delta^{-1}  \right)}\right) \right)^2 =O\left( \frac{nE^2}{\varepsilon^2} \log \left( \delta^{-1} \right) \right)
\end{align}
    ensures that $\frac{1}{2}\left\|\rho(\hat V, \hat m) - \rho(V, m)\right\|_1 \leq \varepsilon,$ with probability at least $1 - \delta$.
\end{thm}
\begin{proof}
Recall that a sample $\hat r$ from a heterodyne measurement is distributed according to a classical Gaussian distribution $\hat r \sim \mathcal{N}\left( m, \frac{V + \mathbb{1}}{2} \right)$ (see the preliminary section). Thus, using Eq.~\eqref{eq:firstmom} from Lemma~\ref{le:estrelative}, we have with probability at least $1 - \delta$:
\begin{align}
    \label{eq:assumpmom}
    \left\|\left(\frac{V+\id}{2}\right)^{-1/2}\left(\hat{m}-m \right)\right\|_2 \le  \frac{\chi}{\sqrt{N}}.
\end{align}
Moreover, by Eq.~\eqref{eq:relative} of the same lemma, the empirical covariance matrix satisfies
\begin{align}
    (1 - \zeta) \frac{V + \id}{2} \le \hat{\Sigma}_h \le (1 + \zeta) \frac{V + \id}{2}.
\end{align}
Hence, the estimator $\hat{V} \coloneqq \frac{2 \hat{\Sigma}_h}{1 - \zeta} - \id$ satisfies
\begin{align}
\label{eq:relerrorhatvH}
    0 \le \hat{V} - V \le \frac{2\zeta}{1 - \zeta}(V + \id).
\end{align}
This implies that $\hat{V}$ is a valid covariance matrix: since $\hat{V} \geq V$ and $V$ is valid, it follows that $\hat{V} + i \Omega \geq V + i \Omega \geq 0$. Thus, $\rho(\hat V,\hat m)$ is a well-defined Gaussian state.
Using the perturbation bound from Theorem~\ref{th:pertboundAPP}, we obtain
\begin{align}
    \|\rho(\hat V,\hat m)-\rho(V,m)\|_1&\leq 
    2\|V^{-1/2}(\hat m-m)\|_2+\frac{1+\sqrt{3}}{4} \Tr((V^{-1}+\hat V^{-1})|\hat V-V|)\\
     \nonumber
    &\leqt{(i)} 2\|V^{-1/2}(\hat m - m)\|_2 
+ \frac{1+\sqrt{3}}{2} \Tr(V^{-1}(\hat V - V)) \\
 \nonumber
    &\leqt{(ii)} 2\left\|V^{-\frac{1}{2}}\left(\frac{V+\id}{2}\right)^{\frac{1}{2}} 
\left(\frac{V+\id}{2}\right)^{-\frac{1}{2}}(\hat m - m)\right\|_2 + (1+\sqrt{3}) \frac{\zeta}{1 - \zeta} 
\Tr(V^{-1}(V + \id)) \\
 \nonumber
    &\leq 2\left\|V^{-1/2}\left(\frac{V+\id}{2}\right)^{1/2}\right\|_\infty 
 \left\|\left(\frac{V+\id}{2}\right)^{-1/2}(\hat m - m)\right\|_2 + (1+\sqrt{3}) \frac{\zeta}{1 - \zeta} 
\left\|\id + V^{-1}\right\|_1 
\\
 \nonumber&\leqt{(iii)}   \left\|V^{-1}+\id\right\|^{1/2}_\infty 
 \frac{\sqrt{2}\chi}{\sqrt{N}} + (1+\sqrt{3}) \frac{\zeta}{1 - \zeta} 
\left(2n + \Tr(V^{-1})\right)
\\
 \nonumber
 &\leqt{(iv)}   \frac{1}{2}\left\|V^{-1}+\id\right\|_1
 \frac{\sqrt{2}\chi}{\sqrt{N}} + (1+\sqrt{3}) \frac{\zeta}{1 - \zeta} 
\left(2n + \Tr(V^{-1})\right)
\\
 \nonumber
 &= \left( \frac{1}{\sqrt{2}}
 \frac{\chi}{\sqrt{N}} + (1+\sqrt{3}) \frac{\zeta}{1 - \zeta}\right)\left(2n + \Tr(V^{-1})\right)
 \\&
 \leqt{(v)}  8.6  \frac{\chi}{\sqrt{N}} \left(2n + \Tr(V^{-1})\right).
\label{eq:lastpr}
\end{align}
In step~(i), we have used that $\hat{V} - V \geq 0$, which implies $\hat{V}^{-1} \leq V^{-1}$ by inequality~\ref{eq:inversebound}, and applied the inequality in Point~\ref{eq:listeqnorm1} from the list of matrix inequalities in the preliminaries. In step~(ii), we invoked Eq.~\eqref{eq:relerrorhatvH} and used inequality~\ref{eq:listeqnorm2} from the same list. Step~(iii) follows from Eq.~\eqref{eq:assumpmom}, while step~(iv) uses inequality~\ref{eq:listeqtracebd} from the preliminaries. In the final step~(v), we have substituted 
\begin{equation}
\zeta \coloneqq \left( \frac{2 \chi}{\sqrt{N}} + \frac{2 \chi^2}{N} \right) 
\end{equation}
and imposed the constraint $\|\rho(\hat V,\hat m)-\rho(V,m)\|_1\leq 2$; under this condition, solving the inequality for $\chi/\sqrt{N}$ yields the claimed bound, which can be verified numerically.
\end{proof}
In practical scenarios, the prefactors appearing in the bounds of Theorem~\ref{th:heterotom} can often be improved, leading to tighter and more refined estimates.


\begin{remark} 
The sample complexity in Theorem~\ref{th:heterotom} depends on \( {\mathrm Tr}(V^{-1}) \), which is generally not known a priori in experimental settings. Typically, one assumes an upper bound on the total energy, \( \Tr(\rho \hat{E}) \le E \), but this is often a loose proxy for \( \Tr(V^{-1}) \), leading to overly conservative estimates for the required number of samples.

To address this, one can directly estimate \( \Tr(V^{-1}) \) from heterodyne data. Given an initial batch of samples, the empirical covariance \( \hat{V} \) can be used to bound \( \Tr(V^{-1}) \) via the inequalities
\begin{align}
    \Tr(V^{-1}) &\le 4E, \\
    \Tr(V^{-1}) &\le \Tr(\hat{V}^{-1}), \\
    \Tr(V^{-1}) &\le \left(1 + \frac{2\zeta}{1 - \zeta} \right) 
    \Tr\left( \left( \hat{V} - \frac{2\zeta}{1 - \zeta} \mathbb{1} \right)^{-1} \right),
    \label{eq:emp_bound}
\end{align}
where the first follows from the energy constraint, the second from \( V^{-1} \le \Omega V \Omega^T \le \Omega \hat{V} \Omega^T \), and the third from transforming Eq.~\ref{eq:relerrorhatvH}. Note that the third bound is valid only when \( \hat{V} - \frac{2\zeta}{1 - \zeta} \mathbb{1} \ge 0 \).
One can check whether this condition is met in practice. For this, we observe that the trace distance guarantee 
\begin{align}
    \varepsilon \coloneqq \left( \frac{1}{2} \frac{\chi}{\sqrt{N}} + (1 + \sqrt{3}) \frac{\zeta}{1 - \zeta} \right) \left( 2n + \Tr(V^{-1}) \right)
\end{align}
requires
\begin{align}
    \frac{\zeta}{1 - \zeta} \le \frac{\varepsilon \lambda_{\min}(V)}{1 + \sqrt{3}} \le \frac{\varepsilon \lambda_{\min}(\hat{V})}{1 + \sqrt{3}}.
\end{align}
Inserting this into Eq.~(\ref{eq:emp_bound}), we obtain the refined bound
\begin{align}
    \Tr(V^{-1}) &\le \left(1 + \frac{2\zeta}{1 - \zeta} \right) 
    \Tr\left( \left( \hat{V} - \frac{2 \varepsilon \lambda_{\min}(\hat{V})}{1 + \sqrt{3}} \mathbb{1} \right)^{-1} \right) \\
     \nonumber
    &\le \left(1 + \frac{2\zeta}{1 - \zeta} \right) 
    \left(1 - \frac{2 \varepsilon}{1 + \sqrt{3}} \right)^{-1} \Tr(\hat{V}^{-1}).
\end{align}
While this expression still contains the non-observable parameter \( \varepsilon \), it highlights that for high-accuracy regimes (\( \varepsilon \ll 1 \)), Eq.~\ref{eq:emp_bound} provides an increasingly sharp estimate of \( \Tr(V^{-1}) \).
This reasoning naturally supports an adaptive tomography strategy: one may collect an initial batch of heterodyne data, estimate \( \Tr(V^{-1}) \) from \( \hat{V} \), and use it to refine the required sample size in subsequent rounds, thereby ensuring the target trace distance is met with high confidence.
\end{remark}

\section{Gaussian state tomography with squeezed inputs: An adaptive and effectively energy-independent algorithm}

Theorem~\ref{th:heterotom} shows that Gaussian state tomography via heterodyne measurements has sample complexity scaling proportionally to \(\|V^{-1}\|_\infty^2\), which in turn grows with the energy of the state. In the following lemma, we show that this energy dependence can be effectively suppressed by adaptively reducing \(\|V^{-1}\|_\infty\).

\begin{lemma}[(Adaptive step to reduce the squeezing)]
\label{le:adaptivestepALG}
Let \(V\) be the covariance matrix of a Gaussian state, and \(\hat{V}\) its estimate as in Theorem~\ref{th:heterotom}, with Williamson decomposition \(\hat{V} = \hat{S} \hat{D} \hat{S}^\intercal\). Define
\begin{align}
    V_{\mathrm{new}} \coloneqq \hat{S}^{-1} V \hat{S}^{-\intercal}.
\end{align}
Then,
\begin{align}
    \|V_{\mathrm{new}}^{-1}\|_{\infty} \leq \sqrt{1+\frac{2\zeta}{1-\zeta}(\|V^{-1}\|_\infty+1)}.
    \label{eq:recV}
\end{align}
In particular, setting \(\zeta \le \frac{1}{4}\) yields $   \|V_{\mathrm{new}}^{-1}\|_{\infty} \leq \sqrt{\frac{5}{3} + \frac{2}{3}\|V^{-1}\|_{\infty}}$ with at least $1-\delta$ probability.
This choice \(\zeta \le \tfrac{1}{4}\) is guaranteed provided the number of heterodyne samples \(N\) satisfies
\begin{align}
   N > N_h(n,\delta), \qquad \text{where} \quad 
   N_h(n,\delta) \coloneqq 80 \left( \sqrt{2n} + \sqrt{2 \log(2/\delta)}  \right)^2 = O(n).
   \label{eq:Nh}
\end{align}

\end{lemma}

\begin{proof}
As in the proof of Theorem~\ref{th:heterotom} (specifically Eq.~\eqref{eq:relerrorhatvH}), the estimate \(\hat{V}\) satisfies
\begin{align}
    V \leq \hat{V} \leq \left[ \id + \frac{2\zeta}{1-\zeta} (V^{-1} + \id) \right] V.
\end{align}
From this, it follows that
\begin{align}\label{eq:ineq_hat_v}
    V \leq \hat{V} \leq \left[ 1 + \frac{2\zeta}{1-\zeta} \big(\|V^{-1}\|_{\infty} + 1 \big) \right] V,
\end{align}
where we have used that \(V^{-1} \leq \|V^{-1}\|_{\infty} \id\), which follows from the eigendecomposition of \(\hat{V}\) (which exists since \(\hat{V} \geq V \geq 0\)), and the fact that if \(A, B \geq 0\) commute, then \(AB \geq 0\).
Now we can use this to bound $V^{-1}$
    \begin{align}
        V^{-1}&\eqt{(i)}V^{-1}\#V^{-1}\\
        \nonumber
        &\leqt{(ii)} (\Omega V\Omega^\intercal)\#V^{-1}\\
         \nonumber
        & \leqt{(iii)} (\Omega \hat{V}\Omega^\intercal)\#\left( \left[1+\frac{2\zeta}{1-\zeta}(\|V^{-1}\|_\infty+1)\right]\hat{V}^{-1}\right)\\
         \nonumber
        &\eqt{(iv)}\sqrt{1+\frac{2\zeta}{1-\zeta}(\|V^{-1}\|_\infty+1)} \,(\Omega \hat{V}\Omega^\intercal)\# \hat{V}^{-1}\\
         \nonumber
        &\eqt{(v)} \sqrt{1+\frac{2\zeta}{1-\zeta}(\|V^{-1}\|_\infty+1)} \hat{S}^{-\intercal}\hat{S}^{-1}\,.
         \nonumber
    \end{align}
Here, in step (i) we applied the definition of the geometric mean (Point~\ref{eq:listeqnorm1} of the preliminary list). In (ii), we have used that \(V^{-1}\le \Omega V\Omega^\intercal\) together with Point~\ref{eq:listeqnorm1}. In (iii), we have used again Point~\ref{eq:listeqnorm1} together with Eq.~\eqref{eq:ineq_hat_v}, and the operator monotonicity of the inverse function \(x\mapsto x^{-1}\) on positive matrices (Point~\ref{eq:inversebound}). In (iv), we exploited Point~\ref{eq:scalar} of the preliminary list. Finally, in (v), we have used that
\begin{align}
        (\Omega \hat{V}\Omega^\intercal)\#\hat{V}^{-1}&\eqt{}(\hat{S}^{-\intercal }\hat{D}\hat{S}^{-1})\#(\hat{S}^{-\intercal }\hat{D}^{-1}\hat{S}^{-1})\eqt{} \hat{S}^{-\intercal}(\hat{D}\#\hat{D}^{-1}) \hat{S}^{-1}\eqt{}\hat{S}^{-\intercal}\hat{S}^{-1}\,,
\end{align}
where we leveraged Point~\ref{eq:congruence} and Point~\ref{eq:commuting} of the preliminary list.

Now, by acting adjointly with \(\hat{S}\) on both sides of the derived inequality and exploiting the definition of \(V_{\mathrm{new}}\), we have that
\begin{align}
    V_{\mathrm{new}}^{-1} = \hat{S}^\intercal V^{-1} \hat{S} \leq \sqrt{1+\frac{2\zeta}{1-\zeta}(\|V^{-1}\|_\infty+1)} \id.
\end{align}
Taking the operator norm completes the proof of Eq.~\eqref{eq:recV}.
Equation~\eqref{eq:Nh} follows from Eq.~\eqref{eq:zetadef}.
Setting \( \zeta = \frac{1}{4} \) gives
\begin{align}
    \frac{1}{4} 
    = \frac{2\chi}{\sqrt{N_h}} + \frac{2\chi^2}{N_h}, \quad \text{where } \chi \coloneqq \sqrt{2n} + \sqrt{2 \log(2/\delta)}.
\end{align}
Letting \( \sqrt{N_h} = x \), this becomes    
\begin{align}
\frac{1}{4} = \frac{2\chi}{x} + \frac{2\chi^2}{x^2} \quad \Rightarrow \quad x = \chi(4 + 2\sqrt{6}), 
\end{align}
so
\begin{align}
N_h = \left( \chi(4 + 2\sqrt{6}) \right)^2\le 80 \chi^2 = 80 \left( \sqrt{2n} + \sqrt{2 \log(2/\delta)} \right)^2.
\end{align}

\end{proof}

The next lemma analyzes the recurrence relation arising from successive adaptive squeezing steps described in Lemma~\ref{le:adaptivestepALG}. Specifically, we consider the sequence \(a_k \coloneqq \|V_k^{-1}\|_{\infty}\), where \(V_k\) is the covariance matrix at step \(k\), and define \(\eta \coloneqq \frac{2\zeta}{1 - \zeta}\). Note that \(\eta \le \frac{2}{3}\) provided that \(\zeta \le \frac{1}{4}\).
\begin{lemma}[(Recurrence relation)]
\label{le:recurrence}
Let \((a_k)_{k \in \mathbb{N}}\) be a sequence of positive real numbers satisfying
\begin{align}
    a_{k+1} \le \sqrt{\frac{5}{3} + \frac{2}{3} a_k} \quad \text{for all } k \in \mathbb{N},
\end{align}
then, for all \(k \ge 0\), as long as \(a_k \ge 2\), it holds that
\begin{align}
    a_k \le \sqrt{\frac{7}{3}} \cdot a_0^{1/2^k}.
\end{align}
In particular, if \(k > \log_2\big( \log_2 a_0 \big)\), then \(a_k \le 2\).
In particular, it holds in general, that if \(a_0 \le 10^{300}\), then \(a_{10} \le 2\) (meaning that $10$ iteration rounds are sufficient for all practical considerations.)
\end{lemma}

\begin{proof}
Now observe that for all $a_k \ge 2$, we have
\begin{align}
    a_{k+1} = \sqrt{\frac{5}{3} + \frac{2}{3} a_k} \le  \sqrt{\frac{5}{3 a_k} + \frac{2}{3} }\sqrt{a_k} \leq \sqrt{\frac{3}{2}} \cdot \sqrt{a_k}\,.
\end{align}
For all $k \ge 1$, this implies the recursive inequality
\begin{align}
   a_k \le \max\left(2,\left(\frac{3}{2}\right)^{\sum_{j=0}^{k-1} 1/2^{j+1}} \cdot a_0^{1/2^k} \right)\le \max\left(2,\frac{3}{2} a_0^{1/2^k}\right)\,.
\end{align}
This implies that if $k\ge \log_2\left(\log_2(a_0)/\log_2(\frac{4}{3})\right)$, then $a_k \le 2$. The numerical value for $a_0=10^{300}$ can be checked with a numerical solver.
\end{proof}

We now have all the necessary building blocks to prove the correctness of our tomography algorithm, which achieves (quasi-) energy-independent performance (i.e., up to a doubly logarithmic factor). The algorithm is detailed in Alg.~2 in the main text. The high-level idea is as follows: we first transform the unknown Gaussian state via squeezing operations into a state with squeezing $\|V^{-1}\|_{\infty}$ bounded by a low constant, perform tomography on this transformed state, and finally apply the inverse squeezing operations to recover an estimate of the original state. More precisely, we perform \(k = \lceil \log_2(\log_2(\|V^{-1}\|_{\infty})) \rceil\) adaptive rounds of the unsqueezing procedure based on heterodyne measurements, described in Lemmas~\ref{le:adaptivestepALG} and~\ref{le:recurrence}. After these \(k\) rounds, we obtain an \emph{unsqueezed} Gaussian state with covariance matrix \(V_k\) satisfying \(\|V_k^{-1}\|_\infty \le 2\), ensuring that the heterodyne tomography described in Theorem~\ref{th:heterotom} can be performed on this obtained state with a sample complexity independent of the covariance matrix of the original state. The final output state is obtained by applying the inverse of the accumulated unsqueezing Gaussian unitaries to the state learned in the last tomography step, yielding a faithful reconstruction of the original state.

\begin{thm}[(Quasi-energy-independent Gaussian state tomography)]
\label{th:indeptom}
Let \(\varepsilon, \delta \in (0,1)\) be the desired accuracy and failure probability, respectively. 
Let \(\rho(V, m)\) be an unknown \(n\)-mode Gaussian state with first moment \(m\) and covariance matrix \(V\). Define
\begin{align}
k \coloneqq \left\lceil \log_2 \left( \log_2 \left( \|V^{-1}\|_{\infty} \right) \right) \right\rceil,
\end{align}
the number of adaptive heterodyne rounds used to unsqueeze the state. Define
\begin{align}
N_h(n, \delta) &= 80 \left( \sqrt{2n} + \sqrt{2 \log(2\delta^{-1})} \right)^2, \\
N_t(n , \varepsilon, \delta) &\coloneqq \left( \frac{21.5}{\varepsilon} n\left(\sqrt{2n} + \sqrt{2 \log \left( 2\delta^{-1} \right)}\right) \right)^2.
\end{align}
{Then,
\begin{align}
N_{\mathrm{tot}} 
&\ge 
k \cdot N_h\left(n, \frac{\delta}{2(k+1)}\right)
+
N_t\left(n,\varepsilon, \frac{\delta}{2}\right) \\
&=
O\left(
\bigl(n+\log\log\log(\|V^{-1}\|_{\infty})\bigr)
\log\log(\|V^{-1}\|_{\infty})
+
\frac{n^3}{\varepsilon^2}
\right),
\nonumber
\end{align}}
copies of \(\rho(V,m)\) are enough to construct a valid estimate \(\rho(\hat{V}, \hat{m})\) such that, with probability at least \(1 - \delta\),
\begin{align}
\frac{1}{2}\|\rho(\hat{V}, \hat{m}) - \rho(V, m)\|_1 \leq \varepsilon.
\end{align}
The algorithm is given in Alg.~2 in the main text. {Moreover, note that in the natural regime \(n \gtrsim \log\log(\|V^{-1}\|_{\infty})\), the sample complexity reduces to $O\!\left(
\frac{n^3}{\varepsilon^2}\right)$.}

\end{thm}

\begin{proof}
The protocol described in Algorithm~2 has two stages:
\begin{enumerate}
    \item \(k = \lceil \log_2(\log_2(\|V^{-1}\|_{\infty})) \rceil\) rounds of adaptive unsqueezing via heterodyne estimation,
    \item followed by heterodyne tomography on the approximately unsqueezed state.
\end{enumerate}


{\paragraph{First stage.} By Lemmas~\ref{le:adaptivestepALG} and~\ref{le:recurrence}, after \(k\) rounds of adaptive-unsqueezing we obtain a state \(\rho^{(k)} = U_{S^{-1}} \rho U_{S^{-1}}^\dagger = \rho(V_k, m_k)\) with covariance matrix satisfying \(\|V_k^{-1}\|_\infty \leq 2\), provided each round uses at least \(N_h(n, \delta/(2(k+1)))\) copies. Each adaptive round fails with probability at most \(\delta/(2(k+1))\), and hence the full adaptive stage fails with probability at most \(k\delta/(2(k+1))\leq \delta/2\).}

Since at least half of the eigenvalues of any valid covariance matrix are \( \ge 1 \)\footnote{Indeed, let \( V = S D S^T \) be a valid covariance matrix with \( S \) symplectic and \( D \ge \mathbb{1} \). Using the Euler decomposition \( S = O_1 Z O_2 \), we get
\begin{align}
V = S D S^T \ge S S^T = O_1 Z^2 O_1^T,
\end{align}
where \( Z = \mathrm{diag}(z_1, \dots, z_n, z_1^{-1}, \dots, z_n^{-1}) \). Hence, by Weyl eigenvalue monotonicity~\cite{bhatia15}, 
\begin{align}
\lambda_j(V) \ge \lambda_j(Z^2).
\end{align}
Since \( Z^2 \) has \( n \) eigenvalues \( z_j^2 \ge 1 \) and \( n \) eigenvalues \( z_j^{-2} \le 1 \), it follows that at least half of the eigenvalues of \( V \) are \( \ge 1 \).},
it follows that at least half of the eigenvalues of \( V_k^{-1} \) are \( \le 1 \). The remaining half are upper bounded by \( \|V_k^{-1}\|_\infty \le 2 \), so
\begin{align}
{\Tr(V_k^{-1}) \le n \cdot 1 + n \cdot 2 = 3n.}
\end{align}

{\paragraph{Second stage.} We now apply heterodyne tomography to \(N_t\) copies of \(\rho^{(k)}\), with failure probability \(\delta/2\). By Theorem~\ref{th:heterotom}, we get $\hat{V}_k$ and $\hat{m}_k$ such that:
\begin{align}
\frac{1}{2}
\|\rho(V_k, m_k) - \rho(\hat{V}_k, \hat{m}_k)\|_1 
\leq 4.3 \frac{\chi_t}{\sqrt{N_t}} \left(2n + \Tr(V_k^{-1})\right) 
\leq 21.5n \frac{\chi_t}{\sqrt{N_t}},
\end{align}
where
\[
\chi_t \coloneqq \sqrt{2n}+\sqrt{2\log(4/\delta)}.
\]
Requiring this to be \(\le \varepsilon\) gives the stated bound on \(N_t\).}

{\paragraph{Failure probability and sample complexity.} The adaptive stage fails with probability at most \(\delta/2\), and the final tomography stage fails with probability at most \(\delta/2\). Thus, by the union bound, the total failure probability is at most \(\delta\).

The total number of samples is
\begin{align}
N_{\mathrm{tot}}
=
k \cdot N_h\left(n,\frac{\delta}{2(k+1)}\right)
+
N_t\left(n,\varepsilon,\frac{\delta}{2}\right).
\end{align}
Using the definitions of \(N_h\) and \(N_t\), this gives
\begin{align}
N_{\mathrm{tot}}
=
O\left(
k\left(n+\log(k+1)+\log(\delta^{-1})\right)
+
\frac{n^2}{\varepsilon^2}
\left(n+\log(\delta^{-1})\right)
\right).
\end{align}
For constant \(\delta\), and since
\[
k=O\left(\log\log(\|V^{-1}\|_\infty)\right),
\qquad
\log(k+1)=O\left(\log\log\log(\|V^{-1}\|_\infty)\right),
\]
we obtain
\begin{align}
N_{\mathrm{tot}}
=
O\left(
\bigl(n+\log\log\log(\|V^{-1}\|_{\infty})\bigr)
\log\log(\|V^{-1}\|_{\infty})
+
\frac{n^3}{\varepsilon^2}
\right).
\end{align}}

\paragraph{Output reconstruction.} Let
\begin{align}
\hat{S} \coloneqq \hat{S}_{k}^{-1} \cdots \hat{S}_1^{-1}, \qquad
\hat{V} \coloneqq \hat{S}^{-1} \hat{V}_k \hat{S}^{-\intercal}, \qquad
\hat{m} \coloneqq \hat{S}^{-1} \hat{m}_k,
\end{align}
where \(\hat{S}_j\) is the symplectic matrix obtained from the Williamson decomposition in the \(j\)-th unsqueezing round, for \(j \in [k]\). These define the final state estimate \(\rho(\hat{V}, \hat{m})\). Using unitary invariance of the trace norm, together with the relations \(V_k = \hat{S} V \hat{S}^\intercal\) and \(m_k = \hat{S} m\), we obtain
\begin{align}
\frac{1}{2}\|\rho(V, m) - \rho(\hat{V}, \hat{m})\|_1
&= \frac{1}{2}\|\rho(V, m) - \rho(\hat{S}^{-1} \hat{V}_k \hat{S}^{-\intercal}, \hat{S}^{-1} \hat{m}_k)\|_1 \\
\nonumber
&= \frac{1}{2}\|\rho(V, m) - U_{\hat{S}}^\dagger \rho(\hat{V}_k, \hat{m}_k) U_{\hat{S}}\|_1 \\
\nonumber 
&= \frac{1}{2}\|U_{\hat{S}} \rho(V, m) U_{\hat{S}}^\dagger - \rho(\hat{V}_k, \hat{m}_k)\|_1 \\
\nonumber 
&= \frac{1}{2}\|\rho(\hat{S} V \hat{S}^\intercal, \hat{S}m) - \rho(\hat{V}_k, \hat{m}_k)\|_1 \\
\nonumber 
&= \frac{1}{2}\|\rho(V_k, m_k) - \rho(\hat{V}_k, \hat{m}_k)\|_1 \le \varepsilon.
\nonumber 
\end{align}
This concludes the proof.
\end{proof}

\begin{example}[(Sample complexity under the energy bound \(E \le 10^{300}\))]
Assume that \(\|V^{-1}\|_{\infty} \le E \le 10^{300}\), where \(E\) denotes the total energy of the quantum state. This upper bound is motivated by conservative physical assumptions on the energy content of any realistic quantum system in our universe. Set \(\delta = 10^{-10}\). Then, Theorem~\ref{th:indeptom} implies that \(k = 10\) rounds of adaptive heterodyne suffice, and the total number of heterodyne measurements is bounded by
{\begin{align}
N_{\mathrm{tot}} 
&\le 10 \cdot 80 \left( \sqrt{2n} + \sqrt{2 \log\left( \tfrac{44}{\delta} \right)} \right)^2 
+ \left( \frac{21.5}{\varepsilon} n \left( \sqrt{2n} + \sqrt{2 \log\left( \tfrac{4}{\delta} \right)} \right) \right)^2 \\
&\le 800 \left( \sqrt{2n} + 7.4 \right)^2 
+ \left( \frac{21.5}{\varepsilon} n \left( \sqrt{2n} + 7.4 \right) \right)^2.
\nonumber 
\end{align}}
Note that this is a conservative upper bound: due to looseness in several intermediate approximations, practical implementations may converge with significantly fewer measurement shots.
\end{example}

\subsection{Measurement scheme using only passive operations and squeezed inputs}
\label{sec:noonline}
Algorithm~2 requires applying an active Gaussian operation to the 
input state, followed by heterodyne detection. However, active Gaussian unitaries are often challenging to implement in practice in optical experiments. Fortunately, the combined effect of such an active transformation followed by heterodyne measurement can be exactly reproduced using only passive Gaussian operations together with a fixed auxiliary squeezed input state---both of which are experimentally accessible with current photonic technologies.

This alternative scheme leverages the well-established equivalence between heterodyne detection and beam-splitter interference followed by homodyne measurements. As reviewed in the preliminaries, a heterodyne measurement can be implemented by sending the state to be measured into one port of a balanced beam splitter, with a vacuum state entering the other port, and then performing homodyne measurements of complementary quadratures---typically position on one output 
arm and momentum on the other. In the variant we consider, the vacuum input is replaced with a suitably chosen squeezed state. This substitution enables one to simulate the effect of any active Gaussian unitary followed by heterodyne detection using only passive elements and an auxiliary squeezed input state.

The following proposition makes this equivalence precise and provides a detailed description of the resulting measurement scheme (see also the main figure in the text).

\begin{prop}[(Passive implementation of heterodyne after active Gaussian unitary)]
\label{prop:passive_heterodyne}
Let \(\rho\) be an \(n\)-mode Gaussian quantum state, and let \(U_S\) be an active Gaussian unitary associated with a symplectic transformation \(S\). Then, heterodyne detection on the transformed state \(U_S \rho U_S^\dagger\) is operationally equivalent with the following.
\begin{enumerate}
    \item Preparing an auxiliary \(n\)-mode Gaussian state \(\sigma\) with covariance matrix \(S^{-1} S^{-T}\), i.e., 
    the squeezed vacuum \(U_S^{-1} \ket{0}\);
    \item Interfering \(\rho\) and \(\sigma\) on a balanced beam splitter;
    \item Performing homodyne detection of \(\hat{x}_1, \dots, \hat{x}_n\) on one output arm and \(\hat{p}_1, \dots, \hat{p}_n\) on the other.
\end{enumerate}
Let \(\mathbf{r}  \in \mathbb{R}^{2n}\) 
denote the outcome \(\mathbf{r}\coloneqq  \sqrt{2}(\textbf{x}_1,\textbf{p}_2)  \in \mathbb{R}^{2n}\). Then \(S \mathbf{r}\) is distributed identically to the outcome of heterodyne detection on \(U_S \rho U_S^\dagger\).
\end{prop}

\begin{proof}
Let \(V\) and \(m\) denote the covariance matrix and first-moment vector of \(\rho\). The transformed state \(U_S \rho U_S^\dagger\) is Gaussian with first moment \(S m\) and covariance matrix \(S V S^T\). A heterodyne measurement on this state yields a sample drawn from the distribution
\begin{align}
\label{eq:hetS}
\mathcal{N}\left(S m ,\, \frac{S V S^T + \id_{2n}}{2} \right).
\end{align}
Now consider the passive scheme. The input to the beam splitter is the product state \(\rho \otimes \sigma\), where \(\sigma\) is a zero-mean Gaussian state with covariance matrix \(S^{-1} S^{-T}\). By Lemma~\ref{le:generalized_het} (generalized heterodyne sampling), the measurement outcome \(\mathbf{r} \in \mathbb{R}^{2n}\) is distributed as
\begin{align}
\label{eq:hetpassive}
\mathbf{r} \sim \mathcal{N} \left( m,\, \frac{V + S^{-1} S^{-T}}{2} \right).
\end{align}
Define the linearly transformed variable \(\tilde{\mathbf{r}} \coloneqq S \mathbf{r}\). Since symplectic transformations are linear and volume-preserving (\(\det S = 1\))~\cite{BUCCO}, the Gaussian nature and normalization of the distribution are preserved. Thus, the resulting distribution is
\begin{align}
\tilde{\mathbf{r}} &\sim \mathcal{N}\left( S m,\, S \left( \frac{V + S^{-1} S^{-T}}{2} \right) S^T \right)  \\
&= \mathcal{N} \left( S m,\, \frac{S V S^T + \id_{2n}}{2} \right).
 \nonumber
\end{align}
This matches exactly the distribution in~\eqref{eq:hetS}, which corresponds to heterodyne detection on the state \(U_S \rho U_S^\dagger\). Therefore, the random variable \(S \mathbf{r}\) obtained from the passive scheme has the correct distribution, completing the proof.
\end{proof}

\subsection{Alternative passive measurement scheme via Euler decomposition}  
We now present an alternative implementation of the passive measurement scheme in Proposition~\ref{prop:passive_heterodyne}, based on the Euler (also known as Bloch-Messiah) decomposition of the symplectic transformation \(S\). Specifically, suppose that $S = O_1 Z O_2,$
where \(O_1\) and \(O_2\) are symplectic orthogonal matrices (i.e., corresponding to passive Gaussian unitaries), and \(Z\) is a diagonal symplectic matrix representing a product of single-mode squeezers. Explicitly, we may write $Z = \mathrm{diag}(z_1, 1/z_1, z_2, 1/z_2, \dots, z_n, 1/z_n),
$ for squeezing parameters \(z_j > 0\).

Then, the passive scheme described in Proposition~\ref{prop:passive_heterodyne} can be reformulated as follows:
\begin{itemize}
    \item Prepare a squeezed vacuum auxiliary state vector  \( U_Z^\dagger \ket{0} \), where \(U_Z\) is the Gaussian unitary implementing the squeezing matrix \(Z\).
    \item Apply the passive Gaussian unitary \(U_{O_2}\) to the input state \(\rho\).
    \item Interfere the resulting state with the auxiliary state at a balanced beam splitter.
    \item Perform homodyne detection of complementary quadratures yielding outcomes $\mathbf{x}_1 \in \mathbb{R}^n$ and $\mathbf{p}_2 \in \mathbb{R}^n$ as before, and let \(\mathbf{r}\coloneqq \sqrt{2}(\mathbf{x}_1,\mathbf{p}_2) \in \mathbb{R}^{2n}\).
    \item Then, the post-processed outcome \(O_1 Z \mathbf{r}\) is distributed according to the same Gaussian distribution as the result of heterodyne measurements on the state \(U_S \rho U_S^\dagger\).
\end{itemize}

\begin{proof}
Let \(V\) and \(m\) denote the covariance matrix and mean vector of the input state \(\rho\). Applying the passive unitary \(U_{O_2}\) maps \(\rho\) to a new Gaussian state with mean \(O_2 m\) and covariance \(O_2 V O_2^T\). The auxiliary state \( U_Z^\dagger \ket{0}\) is a squeezed vacuum with zero mean and covariance matrix \(Z^{-1} Z^{-T}\).
By Lemma~\ref{le:generalized_het}, applying the beam splitter and performing homodyne detection on this setup yields a classical outcome \(\mathbf{r}\) distributed as
\begin{align}
\mathbf{r} \sim \mathcal{N} \left( O_2 m, \, \frac{O_2 V O_2^T + Z^{-1} Z^{-T}}{2} \right).
\end{align}

Now define \(\tilde{\mathbf{r}} \coloneqq O_1 Z \mathbf{r}\). Since linear transformations preserve the Gaussian form of the classical distribution and symplectic orthogonal and squeezing transformations have unit determinant, \(\tilde{\mathbf{r}}\) is also Gaussian distributed, with
\begin{align}
\tilde{\mathbf{r}} &\sim \mathcal{N} \left( O_1 Z O_2 m, \, \frac{1}{2} O_1 Z \left( O_2 V O_2^T + Z^{-1} Z^{-T} \right) Z^T O_1^T \right) \\
&= \mathcal{N} \left( S m, \, \frac{1}{2} \left( S V S^T + \id_{2n} \right) \right).
\nonumber 
\end{align}
This matches exactly the outcome distribution of a heterodyne measurement on the transformed state \(U_S \rho U_S^\dagger\), completing the proof.
\end{proof}

\noindent
Note that if the passive transformation \(O_2\) is instead applied to the squeezed auxiliary state prior to the beam splitter (rather than on $\rho$), the scheme becomes identical to the one described in Proposition~\ref{prop:passive_heterodyne}.

\section{Gaussian state tomography with access to the transposed state: A fully energy-independent protocol}

The tomography protocol introduced in the previous section achieves a sample complexity that depends only very mildly on the total energy of the state---specifically, doubly logarithmically on \(\|V^{-1}\|_\infty\). In all physically relevant scenarios, this dependence is negligible, and the protocol can be considered nearly energy-independent.

However, if access to the transpose \(\rho^T\) of the unknown state \(\rho\) is available, one can go a step further: we can construct a tomography algorithm with completely constant sample complexity, fully independent of the energy of the state (or \(\|V^{-1}\|_\infty\)). The key technical insight enabling this improvement is captured by the following lemma.

\begin{lemma}[(Scheme with access to the transpose)]
\label{lem:transpose_sampling}
Let \(\rho\) be an \(n\)-mode Gaussian quantum state with first-moment vector \(m \in \mathbb{R}^{2n}\) and covariance matrix \(V \in \mathbb{R}^{2n \times 2n}\). Then, the following passive measurement scheme acting on \(\rho \otimes \rho^T\) yields a classical sample distributed as \(\mathcal{N}(2 m, V)\):
\begin{enumerate}
    \item Interfere \(\rho\) and \(\rho^T\) at a balanced beam splitter.
    \item Perform homodyne measurements of the position quadratures \(\hat{x}_1, \dots, \hat{x}_n\) on one output arm, and of the momentum quadratures \(\hat{p}_1, \dots, \hat{p}_n\) on the other, yielding outcomes $\mathbf{x}_1 \in \mathbb{R}^n$ and $\mathbf{p}_2 \in \mathbb{R}^n$ as before.
    \item Let \(\mathbf{r} \coloneqq \sqrt{2}(\mathbf{x}_1,\mathbf{p}_2) \in \mathbb{R}^{2n}\) denote the classical outcome. Then, $\mathbf{r} \sim \mathcal{N}\left( 2m,\, V \right)$.
\end{enumerate}
\end{lemma}
\begin{proof}
The transpose \(\rho^T\) is again a Gaussian state (as can be verified from the Gibbs state definition involving quadratic Hamiltonians). It has the same covariance matrix as \(\rho\), but with a sign flip in the position-momentum correlations. Specifically, if the covariance matrix of \(\rho\) is written in block form as
\begin{align}
    V(\rho) \cong
    \begin{pmatrix}
        V_{x,x} & V_{x,p} \\
        V_{x,p}^T & V_{p,p}
    \end{pmatrix},
\end{align}
then the covariance matrix of the transposed state \(\rho^T\) becomes
\begin{align}
    V(\rho^T) \cong
    \begin{pmatrix}
        V_{x,x} & -V_{x,p} \\
        -V_{x,p}^T & V_{p,p}
    \end{pmatrix},
\end{align}
where \(\cong\) indicates equality up to a reshuffling of quadratures (e.g., grouping all \(x_j\) and all \(p_j\) together). This follows from the fact that transposition leaves the position operators unchanged, \(x_j^T = x_j\), while flipping the sign of the momenta, \(p_j^T = -p_j\). For instance, for any observable \(x_j p_i\) with \(j \neq i\), one has
\begin{align}
    \Tr(x_j p_i \rho)
    &= \Tr\left((x_j p_i \rho)^T\right)
    = \Tr\left(p_i^T x_j^T \rho^T\right)
    = \Tr\left((-p_i) x_j \rho^T\right) \nonumber \\
    &= -\Tr(x_j p_i \rho^T),
\end{align}
demonstrating that the cross-correlations between position and momentum operators acquire a minus sign under transposition, while \(V_{x,x}\) and \(V_{p,p}\) remain unchanged.

Equivalently, the full covariance matrix of \(\rho^T\) can be written as $V(\rho^T) = (\mathbb{1}_n \otimes Z) V (\mathbb{1}_n \otimes Z)$,
where \(\mathbb{1}_n \otimes Z\) flips the signs of all momentum quadratures, with $Z=\mathrm{diag}(1,-1)$.
The first moment of $\rho^T$ is  $(\mathbb{1}_n \otimes Z) m(\rho)$.
Now consider the joint state \(\rho \otimes \rho^T\), which is a \(2n\)-mode Gaussian state with first-moment vector
\begin{equation}
\begin{pmatrix}
    m(\rho) \\
    (\mathbb{1}_n \otimes Z) m(\rho)
\end{pmatrix},
\quad \text{and covariance matrix} \quad
\begin{pmatrix}
    V(\rho) & 0 \\
    0 & (\mathbb{1}_n \otimes Z) V(\rho) (\mathbb{1}_n \otimes Z)
\end{pmatrix}.
\end{equation}
Interfering \(\rho\) and \(\rho^T\) on a balanced beam splitter and performing homodyne measurements of complementary quadratures (position on one output arm, momentum on the other) corresponds to the generalized heterodyne setup described in Lemma~\ref{le:generalized_het}. Applying that result, the classical outcome \(\mathbf{r} \in \mathbb{R}^{2n}\) is distributed as
\begin{align}
    \mathbf{r} \sim \mathcal{N} \left( m(\rho)+ (\mathbb{1}_n \otimes Z) m(\rho^T),\, \frac{V(\rho) + (\mathbb{1}_n \otimes Z) V(\rho) (\mathbb{1}_n \otimes Z)}{2} \right)= \mathcal{N} \left( 2 m(\rho),\, V(\rho) \right).
\end{align}

\end{proof}
 
\begin{thm}[(Energy-independent Gaussian state tomography with access to the transposed state)]
\label{th:indeptomTRAPP}
Let \(\rho \coloneqq \rho(V, m)\) be an unknown \(n\)-mode Gaussian state with first-moment vector \(m \in \mathbb{R}^{2n}\) and covariance matrix \(V \in \mathbb{R}^{2n \times 2n}\).

Let \(\{\hat{r}_i\}_{i \in [N]}\) be \(N\) i.i.d. samples obtained from a generalized heterodyne measurement on the state \(\rho(V, m) \otimes \rho(V, m)^T\), following the scheme described in Lemma~\ref{lem:transpose_sampling}. Define the estimators
\begin{align}
    \hat{m} &\coloneqq   \, \frac{\hat{\mu}_h}{2}, \quad 
    \hat{V} \coloneqq \frac{\hat{\Sigma}_h}{1-\zeta},
\end{align}
where $\hat{\mu}_h \coloneqq \frac{1}{N} \sum_{i \in [N]} \hat{r}_i$, and 
$\hat{\Sigma}_h \coloneqq \frac{1}{N} \sum_{i \in [N]} (\hat{r}_i - \hat{\mu}_h)(\hat{r}_i - \hat{\mu}_h)^\top$.
Then, with probability at least \(1 - \delta\), it holds that
\begin{align}
\frac{1}{2}\|\rho(\hat{V}, \hat{m}) - \rho(V, m)\|_1 
\leq 8.7 n \frac{\chi}{\sqrt{N}},
\end{align}
where \(\zeta \coloneqq \frac{2\chi}{\sqrt{N}} + \frac{2\chi^2}{N}\) and $\chi \coloneqq \sqrt{2n} + \sqrt{2 \log \left( \frac{2}{\delta} \right)}.$
In particular, for any \(\varepsilon \in (0,1)\), it suffices to choose $N = O\left( \frac{n^3}{\varepsilon^2} \log \left( \delta^{-1} \right) \right)$ in order to guarantee that \(\frac12\|\rho(\hat{V}, \hat{m}) - \rho(V, m)\|_1 \leq \varepsilon\) with probability at least \(1 - \delta\).
\end{thm}

\begin{proof}
The proof follows similar steps to Theorem~\ref{th:heterotom} (but is simpler). By Lemma~\ref{lem:transpose_sampling}, a sample \(\hat{r}\) obtained from homodyne measurements on \(\rho \otimes \rho^{T}\) after applying a beam splitter is distributed as
\begin{equation}
\hat{r} \sim \mathcal{N}\left( 2m, V \right).
\end{equation}
Using Lemma~\ref{le:estrelative}, Eq.~\eqref{eq:firstmom}, we can build an estimator \(2\hat{m}\) for \(2m\), such that with probability at least \(1 - \delta\),
\begin{align}
\label{eq:assumpmom2}
\left\|V^{-1/2}(2\hat{m} - 2m)\right\|_2 \le \frac{\chi}{\sqrt{N}}.
\end{align}
In addition, by Eq.~\eqref{eq:relative}, the empirical covariance \(\hat{V}\) satisfies
\begin{align}
(1 - \zeta) V \le \hat{\Sigma}_h \le (1 + \zeta) V,
\end{align}
where \(\zeta \coloneqq \frac{2\chi}{\sqrt{N}} + \frac{2\chi^2}{N}\).
Hence, the estimator $\hat{V} \coloneqq \frac{\hat{\Sigma}_h}{1 - \zeta} $ satisfies
\begin{align}
\label{eq:relerrorhatvT}
    0 \le \hat{V} - V \le \frac{2\zeta}{1 - \zeta}V.
\end{align}
Thus, $\hat{V}$ is a valid covariance matrix. 
Applying the perturbation bound from Theorem~\ref{th:pertboundAPP}, we now obtain
\begin{align}
\|\rho(\hat{V}, \hat{m}) - \rho(V, m)\|_1 
&\le 2\|V^{-1/2}(\hat{m} - m)\|_2 
+ \frac{1 + \sqrt{3}}{4}\Tr((V^{-1}+\hat V^{-1})(\hat{V} - V)) \\
 \nonumber
&\le 2\|V^{-1/2}(\hat{m} - m)\|_2 
+ \frac{1 + \sqrt{3}}{2} \Tr(V^{-1}(\hat{V} - V)) \\
 \nonumber
&\leq \frac{\chi}{\sqrt{N}} 
+ \frac{1 + \sqrt{3}}{2} \times \frac{2\zeta}{1-\zeta}\Tr(V^{-1}V)\\
 \nonumber
 &= \frac{\chi}{\sqrt{N}} 
+ 2n (1 + \sqrt{3})\frac{\zeta}{1-\zeta}\\
 \nonumber
 &\leq 17.4 n\frac{\chi}{\sqrt{N}},
\nonumber
 \end{align} 
where in the last step we used a numerical solver (along with the fact that the one-norm difference between two quantum states is always $\le 2$).


\end{proof}

\end{document}